\documentclass[aps,
    physrev,
    twocolumn,
    groupedaddress,
    showkeys]{revtex4-2}
\usepackage{hhline}
\usepackage[table]{xcolor}
\usepackage{graphicx} % Required for inserting images
\usepackage[english]{babel}
\usepackage[top=2.5cm, bottom=2.5cm, left=2.5cm, right=2.5cm]{geometry} % setting page margins
\usepackage[utf8]{inputenc}
\usepackage[T1]{fontenc}
\usepackage{setspace}
\usepackage{float}
\usepackage{microtype}
\usepackage{amsmath}
\usepackage{amssymb}
\usepackage{amsthm}
\usepackage{physics}
\usepackage{dsfont}
\usepackage{bbm} % Provides \mathbbm for blackboard bold numbers
\usepackage{csquotes}
\usepackage{dcolumn}% Align table columns on decimal point
\usepackage{bm}% bold math
\usepackage{hyperref}
\newtheorem{mytheorem}{Theorem}[section]
\newtheorem{myprop}{Proposition}[section]

\newcommand{\denop}{\mathcal{D}}
\newcommand{\hs}{\mathcal{H}}
\newcommand{\linop}{\mathcal{L}}
\DeclareMathOperator{\mytr}{Tr}

\newcommand{\sumkl}{\sum_{k,l = 0}^{d-1}}
\newcommand{\sumj}{\sum_{j = 0}^{d-1}}

\newcommand{\Pkl}{P_{k,l}}
\newcommand{\Okl}{\Omega_{k,l}}

\newcommand{\ckl}{c_{k,l}}

\newcommand{\Wkl}{W_{k,l}}

\newcommand{\id}{\mathbbm{1}}

\newcommand{\R}{\mathbbm{R}}
\newcommand{\C}{\mathbbm{C}}

\newcommand{\ZZ}{\mathbbm{Z}}

\newcommand{\PPT}{\mathrm{PPT}}
\newcommand{\NPT}{\mathrm{NPT}}
\newcommand{\SEP}{\mathrm{SEP}}

\newcommand{\ENT}{\mathrm{ENT}}

\begin{document}

\title{\textbf{Group-Theoretic Perspective on the PPT and Realignment Criteria in the Magic Simplex for Bipartite Qutrits} 
}% 

\author{Tobias C. Sutter}
 \email{Contact author: tobias.christoph.sutter@univie.ac.at}
\author{Christopher Popp}%
 \email{Contact author: christopher.popp@univie.ac.at}
\author{Beatrix C. Hiesmayr}
 \email{Contact author: beatrix.hiesmayr@univie.ac.at}
\affiliation{%
 University of Vienna, Faculty of Physics, Währingerstrasse 17, 1090 Vienna.
}%

\date{\today}% It is always \today, today,

\begin{abstract}
Entanglement is a key feature in many quantum technologies, including secure communication protocols and quantum computing.
However, detecting it in mixed quantum states remains a challenging task.
While the positive partial transposition (PPT) and computable cross-norm/realignment criteria are well-established tools for entanglement detection in general, and are especially effective in Bell-diagonal states, their connection to the underlying group structure of this state family has not been fully explored.
In this work, we analyze the PPT and realignment criteria for Bell-diagonal states from a group-theoretic point of view.
Our results demonstrate that the group structure of Bell-diagonal states provides a clear framework for analyzing and computing these two entanglement detection criteria, thereby highlighting the connection between entanglement and group structure.
This unified perspective offers new insights into the mathematical and physical properties of entanglement in structured quantum systems and ties the PPT and realignment criteria for Bell-diagonal states to experimental procedures.
\end{abstract}


\maketitle

\section{Introduction}

Many quantum technologies, including quantum computing \cite{preskill_quantum_2018}, cryptography \cite{zapatero_advances_2023}, and teleportation \cite{bennett_teleporting_1993}, rely on entanglement as a fundamental resource.
As these technologies advance, a deeper understanding of entanglement becomes increasingly essential for creating and manipulating resource states.
Still, one of the major open problems in quantum information theory is determining whether a given bipartite quantum state is separable or entangled \cite{horodecki_quantum_2009}.
For mixed quantum states, which arise naturally in experiments due to decoherence, this so-called \emph{separability problem} is provably NP hard when considering systems of arbitrary dimensions \cite{gurvits_classical_2003, gharibian_strong_2010}.
Nevertheless, efficient solutions may exist for fixed local dimensions or particular state families.

In the special case of pure states, and hence especially for bipartite pure states of dimension $d_A\times d_B$, the separability problem has a conclusive and efficiently computable answer \cite{guhne_entanglement_2009, bennett_concentrating_1996}.
Furthermore, it is also solved for general qubit-qubit ($d_A=d_B=2$) and qubit-qutrit ($d_A=d_B-1=2$) systems, because the positive partial transposition (PPT) criterion \cite{peres_separability_1996, horodecki_separability_1996, woronowicz_positive_1976} is necessary and sufficient for entanglement in these dimensions.
However, for higher dimensions ($d_Ad_B>6$), the PPT criterion ceases to be necessary for entanglement, leading to the notion of PPT-entangled states \cite{horodecki_separability_1997}.
Physically, such states are relevant as they cannot be distilled to maximally entangled form using local operations and classical communication \cite{horodecki_mixed-state_1998}.
More generally, there exist multiple other efficiently computable criteria for entanglement detection that are necessary \emph{or} sufficient (cf. \cite{horodecki_quantum_2009, guhne_entanglement_2009, chruscinski_entanglement_2014, hiesmayr_bipartite_2025}).
They aim to approximate a solution to the separability problem in higher dimensions, albeit in a one-sided manner.
One prominent example is the computable cross-norm/realignment criterion \cite{chen_matrix_2003, rudolph_further_2005, chen_test_2004}.

Besides the possibility of fixing the dimensions of the subsystems, a different way to obtain further results is to consider specific state families that bear desirable properties while still being experimentally relevant.
One such family of bipartite quantum states is Bell-diagonal states \cite{ranade_asymptotic_2006, baumgartner_special_2007}.
The physical motivation for investigating these states with respect to their entanglement is that they naturally arise when considering (generalized) Pauli errors.
Thus, they are intimately connected to applications like error correction and entanglement distillation \cite{bennett_mixed-state_1996} (see also Ref.~\cite{popp_novel_2025} for a rigorous treatment of the usefulness of Bell-diagonal states in stabilizer-based entanglement distillation protocols).

Regarding their entanglement structure, Ref.~\cite{popp_almost_2022, popp_special_2024} found that the set of Bell-diagonal qutrit-qutrit states ($d_A = d_B = 3$) contains a significant share of PPT-entanglement, namely about 6\%.
This means that the PPT criterion is relatively ineffective for detecting entanglement in this subclass of states, being successful only for roughly 60\% of all states (see Fig.~\ref{fig:shares_ppt_and_realignment}).
In contrast, the realignment criterion appears to be particularly effective in this dimension, detecting around 62\% of all states as entangled.
Importantly, this involves about 10\% of the states where the PPT criterion fails.
Further numerical results in Ref.~\cite{popp_comparing_2023} suggest that the group structure of Bell-diagonal states crucially affects the entanglement class.
Leveraging this connection of entanglement and group structure of Bell-diagonal states, the recently developed stabilizer distillation protocol FIMAX \cite{popp_novel_2025} achieves superior efficiency compared to other established recurrence protocols, especially for low-fidelity states \cite{popp_low-fidelity_2025}.
This leads to the central question of how the group properties of Bell-diagonal states are related to the performance of the PPT and realignment criteria.
Our contribution investigates this relationship analytically.

\begin{figure}
    \includegraphics[width=0.9\linewidth]{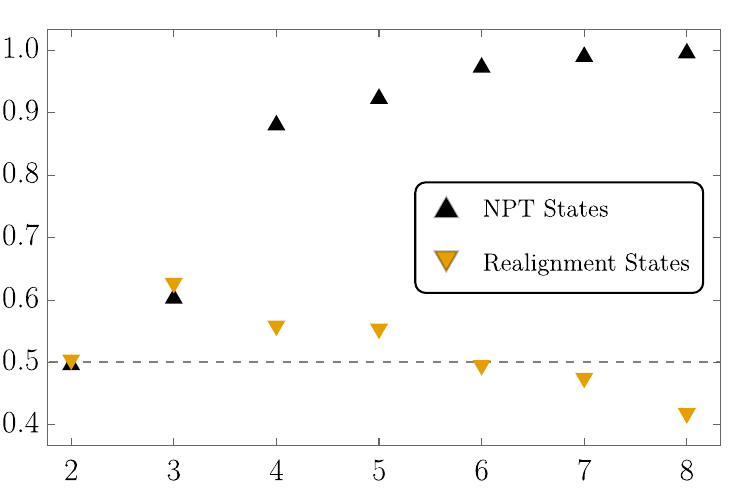}
    \caption{
    Percentages of Bell-diagonal states detected by the PPT and realignment criterion for $d=2,\dots, 8$.
    For each $d$, a total of $10^5$ uniformly sampled states (following a $\mathrm{Dirichlet}(1,\ldots,1)$ distribution on the simplex \eqref{eq:magic_simplex} below) were generated using the Julia package \cite{popp_belldiagonalqudits_2023}.
    For qubits ($d=2$), the PPT criterion is equivalent to the realignment criterion, both detecting around 50\% of states as entangled, which are all existing entangled states in this state family.
    For $d=3$, the realignment criterion is more effective (62.3\%) in detecting entanglement than the PPT criterion (60.7\%).
    In higher dimensions, the vast majority of states are NPT entangled.
    In contrast, the number of entangled states detected by the realignment criterion decreases with increasing dimension $d>3$.
    }
    \label{fig:shares_ppt_and_realignment}
\end{figure}

The work is organized as follows.
Sec.~\ref{sec:PPT_and_realignment} formalizes the PPT and realignment criteria.
In Sec.~\ref{sec:BDS}, we introduce Bell-diagonal states with a focus on their group properties inherited by the Weyl matrices.
We furthermore compute their Bloch vectors and find that they carry information about subgroups of the Weyl-Heisenberg group.
Based on this, we derive an explicit version of the realignment criterion in Sec.~\ref{sec:Realignment_for_BDS}, unveiling its dependence on the subgroups for Bell-diagonal qutrits in Sec.~\ref{sec:Realignment_qutrit_BDS}.
Sec.~\ref{sec:PPT_for_qutrit_BDS} shows that the PPT criterion for Bell-diagonal qutrits can also be understood in terms of the subgroups of the Weyl-Heisenberg group.
We furthermore derive observables that detect NPT states and prove a statement that can guide the search for PPT-entangled Bell-diagonal qutrits.
Finally, Sec.~\ref{sec:conclusion} summarizes the results and gives an outlook on future research directions.

%%%%%%%%%%%%%%%%%%%%%%%%%%%%%%%%%%%%%%%%%%%%%%%%%%%%%%%%%%%%%%%%%%%%%%%%%%%%%%%%%%%%%%%%%%%%%
\section{PPT and Realignment Criteria} \label{sec:PPT_and_realignment}

Let $\linop(\hs)$ be the set of linear operators mapping the Hilbert space $\hs$ into itself, and $\denop(\hs) := \{\rho \in \linop(\hs) \, | \, \rho\geq0, \mytr(\rho) =1\}$ be the set of density matrices on $\hs$.
For bipartite systems, we set $\hs=\hs_A\otimes\hs_B$.
A state $\rho\in\denop(\hs)$ is called separable if it can be written as $\rho=\sum_{i} p_i \, \rho_i^A \otimes \rho_i^B$ for some set $\{p_i , \rho_i^A, \rho_i^B\}$, where $p_i \geq 0$ and $\rho_i^{A/B} \in \denop(\hs_{A/B})$ for all $i$, and $\sum_i p_i =1$; otherwise, $\rho$ is called entangled.
Using these notions, $\denop(\hs)$ can be partitioned into the sets SEP and ENT that denote separable and entangled states, respectively.
Hence, $\SEP \cup \ENT = \denop(\hs_A\otimes\hs_B)$.
Any state $\rho\in\denop(\hs_A\otimes\hs_B)$ can be written as
\begin{align} \label{eq:rho_general}
    \rho = \sum_{i,j=0}^{d_A-1} \sum_{k,l=0}^{d_B-1} \rho_{ik,jl} |i\rangle\langle j|\otimes |k\rangle\langle l| \;,
\end{align}
where $\dim(\hs_A) = d_A$ and $\dim(\hs_B) = d_B$, and $\{|i\rangle\langle j|\}_{i,j=0}^{d_A-1}$ and $\{|k\rangle\langle l|\}_{k,l=0}^{d_B-1}$ constitute a basis of $\denop(\hs_A)$ and $\denop(\hs_B)$, respectively.
Throughout, the transposition map $T:\linop(\hs)\rightarrow\linop(\hs)$ is defined with respect to the computational basis $\{|i\rangle\}_{i=0}^{d-1}$ of $\hs$.

We define the partial transpose acting on the second subsystem, denoted by a superscript $\Gamma$, by its action on the basis states,
\begin{align}
    (|i\rangle\langle j|\otimes |k\rangle\langle l|)^\Gamma &:= (\id_A \otimes T_B)(|i\rangle\langle j|\otimes |k\rangle\langle l|) \notag\\
    &=|i\rangle\langle j|\otimes |l\rangle\langle k| \;, \label{eq:partial_transpose}
\end{align}
where $\id_A$ is the identity map on $\linop(\hs_A)$ and $T_B$ is the transposition map on $\linop(\hs_B)$.
The action of $\Gamma$ extends to general states of the form \eqref{eq:rho_general} by linearity.

The positive partial transposition criterion (also known as PPT or Peres-Horodecki criterion) is a powerful sufficient criterion for entanglement:
\begin{mytheorem}[PPT Criterion \cite{horodecki_separability_1996, peres_separability_1996}] \label{thm:PPT_crit}
    Let $\rho\in\denop(\hs_A\otimes\hs_B)$, and $\rho^\Gamma$ be the partial transpose of $\rho$.
    Then
    \begin{align} \label{eq:PPT_crit}
        \rho^\Gamma \ngeq 0 \Longrightarrow \rho\in\ENT \;.
    \end{align}
\end{mytheorem}

If $\rho^\Gamma \ngeq 0$, i.e., it is not positive semidefinite, we say it is NPT entangled and write $\rho\in\NPT$, and otherwise $\rho\in\PPT$.
This admits another partition of the space of bipartite density operators, namely $\PPT \cup \NPT = \denop(\hs_A \otimes \hs_B)$.
By Thm.~\ref{thm:PPT_crit} it holds that $\rho\in\NPT\Longrightarrow\rho\in\ENT$, but the converse is only true for $d_A d_B \leq 6$ \cite{horodecki_separability_1996, horodecki_separability_1997}.
Hence, $\SEP \neq \PPT$ for higher-dimensional states, which crucially includes qutrit-qutrit systems with $d_A = d_B = 3$.

%%%%%%%%%%%%%%%%%%%%%%%

The realignment criterion (also known as cross-norm criterion) is another sufficient criterion for entanglement.
It is based on the realignment (or shuffle) operation on $\linop(\hs_A\otimes\hs_B)$, which is defined by its action on the basis states as
\begin{align} \label{eq:realignment_operator}
    \hat{R}(|i\rangle\langle j|\otimes |k\rangle\langle l|)
    := |i\rangle\langle k|\otimes |j\rangle\langle l| \;,
\end{align}
and extends to general $\rho\in\denop(\hs_A\otimes\hs_B)$ by linearity.
The computable cross-norm/realignment criterion (hereafter only referred to as realignment criterion) reads:
\begin{mytheorem}[Realignment Criterion \cite{chen_matrix_2003, rudolph_further_2005}] \label{thm:realignment}
    Let $\rho\in\denop(\hs_A\otimes \hs_B)$ and define the realigned state $\rho_R := \hat{R}(\rho)$.
    Then
    \begin{align} \label{eq:realignment_crit}
        \mytr\left( \sqrt{\rho_R^\dagger \rho_R} \right) >1 \Longrightarrow \rho\in\ENT \;.
    \end{align}
\end{mytheorem}
Note that $\mytr( \sqrt{\rho_R^\dagger \rho_R} )$ corresponds to the trace norm of $\rho_R$ (i.e., the sum of its singular values).
Computing the norm scales as $\mathcal{O}(\min(d_A^4d_B^2, d_A^2 d_B^4))$ because $\rho_R$ is a matrix of size $d_A^2\times d_B^2$ \footnote{
Sec.~8.6 of Ref.~\cite{golub_matrix_2013} discusses various algorithms with a scaling of $\mathcal{O}(mn^2)$ for computing the singular values $\{\sigma_i(M)\}_{i=1}^{n}$ of a (possibly dense) matrix $M\in \C^{m\times n}$ with $m\geq n$.
Applied to $\rho_R\in\C^{d_A^2\times d_B^2}$ this yields: If $d_A \geq d_B$, we can compute $\sigma_i(\rho_R)$ for $1\leq i\leq d_B^2$ in $\mathcal{O}(d_A^2 d_B^4)$; if $d_A \leq d_B$, we can compute $\sigma_i(\rho_R^T) = \sigma_i(\rho_R)$ for $1\leq i\leq d_A^2$ in $\mathcal{O}(d_B^2 d_A^4)$ as $\rho_R^T \in \C^{d_B^2\times d_A^2}$.
In either case, computing the trace norm of $\rho_R$ scales as $\mathcal{O}(\min(d_A^4 d_B^2, d_A^2 d_B^4))$.
}.
The realignment criterion detects some PPT entanglement, even though it does not detect all NPT entangled states.
Hence, it is neither stronger nor weaker than the PPT criterion.

Furthermore, the realignment criterion can be expressed in terms of the correlation tensor $T_{a,b} = \mytr(\rho \, (G_a^A \otimes G_b^B))$, where $\{G_a^A\}_{a=0}^{d_A^2-1}$ and $\{G_b^B\}_{b=0}^{d_B^2-1}$ form orthonormal matrix bases of $\denop(\hs_A)$ and $\denop(\hs_B)$, respectively.
The normalization condition is $\mytr(G_a^A (G_{a'}^A)^\dagger) = \delta_{a,a'}$ and $\mytr(G_b^B (G_{b'}^B)^\dagger) = \delta_{b,b'}$.
With this, Thm.~\ref{thm:realignment} can be expressed as \cite{sarbicki_family_2020}
\begin{align} \label{eq:realignment_correlation_tensor}
    \Vert T \Vert_{\mytr} > 1 \Longrightarrow \rho \in \ENT \;,
\end{align}
where $\Vert T \Vert_{\mytr} = \mytr(\sqrt{TT^\dagger})$ is the trace norm of $T$.

%%%%%%%%%%%%%%%%%%%%%%%%%%%%

We conclude this section by noting three similarities between these two entanglement criteria.
(i) Both are based on permutations of the basis states.
(ii) The defining operations represent generally non-physical transformations of the state $\rho$: The transposition is a positive but not completely positive map, and realignment does not preserve hermiticity.
(iii) Among all entanglement criteria for bipartite systems that are based on permutations of the basis states and taking the trace norm, the only two nontrivial, combinatorially independent criteria are the PPT and the realignment criteria \cite{Horodecki_SeparabilityMixedQuantum_2006, wocjan_characterization_2005}.
Hence, specific other permutations of the basis elements, similar to \eqref{eq:partial_transpose} and \eqref{eq:realignment_operator}, recover the criteria:
First, both are trivially invariant under exchanging $\hs_A$ and $\hs_B$, i.e., mapping $\rho\mapsto \mathds{F} \, \rho \, \mathds{F}^\dagger$ where $\mathds{F} = \sum_{i=0}^{d_B-1} \sum_{j=0}^{d_A-1} |i\rangle\langle j| \otimes |j\rangle\langle i|$.
Furthermore, for the PPT criterion, we can equivalently take the partial transposition on the first subsystem $\rho^{T_A}$, acting as $(|i\rangle \langle j| \otimes |k\rangle\langle l|)^{T_A} = |j\rangle\langle i|\otimes |k\rangle\langle l|$ on the basis states.
This leaves the criterion unchanged as the spectrum of $\rho^\Gamma = (\rho^{T_A})^T$ is unaffected by transposition.
Regarding the realignment criterion, we note that $\mathds{F}_A \hat{R}(\rho)$, $\hat{R}(\rho) \mathds{F}_B$, $\mathds{F}_A\hat{R}(\rho) \mathds{F}_B$, $\hat{R}(\rho)^T$, $\mathds{F}_A \hat{R}(\rho)^T$, $\hat{R}(\rho)^T \mathds{F}_B$, $\mathds{F}_A\hat{R}(\rho)^T \mathds{F}_B$ define different permutations of the basis states, where $\mathds{F}_X = \sum_{i=0}^{d_X-1} \sum_{j=0}^{d_X-1} |i\rangle\langle j| \otimes |j\rangle\langle i|$ with $X\in\{A,B\}$ are the flip operators on $\hs_A$ and $\hs_B$.
As the singular values of $\hat{R}(\rho)$ are invariant under taking the transposition or multiplying it by unitaries from the left or right, these permutations again lead to the realignment criterion.

%%%%%%%%%%%%%%%%%%%%%%%%%%%%%%%%%%%%%%%%%%%%%%%%%%%%%%%%%%%%%%%%%%%%%%%%%%%%%%%%%%%%%%%%%%%%%
\section{Bell-Diagonal States with Weyl-Structure} \label{sec:BDS}

Let $\hs = \hs_A \otimes \hs_B$ with $\dim(\hs_A) = \dim(\hs_B) = d$ be the Hilbert space of a bipartite qudit system with equal local dimensions.
Maximally entangled states $|\psi\rangle\in \hs$, i.e., pure states satisfying $\mytr_A(|\psi\rangle\langle\psi|) = \mytr_B(|\psi\rangle\langle\psi|) = \frac{1}{d} \id_d$, are called Bell states on $\hs$.
A set $\{|\psi_i\rangle\}_{i=0}^{d^2-1}$ of $d^2$ orthonormal Bell states on $\hs$ is referred to as a complete Bell basis of $\hs$.
A particularly useful one is generated using the $d$-dimensional unitary Weyl-Heisenberg operators
\begin{align} \label{eq:Wkl}
    \Wkl = \sumj \omega^{jk} \, |j\rangle\langle j+l|
\end{align}
where $k,l\in\{0,1,\dots,d-1\}$ and $\omega := e^{\frac{2\pi i}{d}}$, and, throughout this paper, variables taking values in $\{0,1,\dots, d-1\}$ are understood as elements of $\ZZ_d:= \ZZ/d\ZZ$, i.e., the ring of integers with addition and multiplication modulo $d$.
The matrices \eqref{eq:Wkl} satisfy the so-called Weyl relations
\begin{gather}
    \begin{aligned} \label{eq:Weyl_relations}
        W_{i,j} W_{k,l} &= \omega^{jk} \, W_{i+k,j+l} \;,\\
        \Wkl^\dagger &= \omega^{kl} W_{-k,-l} = \Wkl^{-1} \;, \\
        \Wkl^\ast &= W_{-k,l} \;, \\
        \Wkl^T &= \omega^{-kl} W_{k,-l} \;,
    \end{aligned}
\end{gather}
where $\Wkl^\ast$, $\Wkl^T$, and $\Wkl^\dagger$ denote the complex conjugate, transpose, and conjugate transpose of $\Wkl$, respectively.
Writing $\Wkl = Z^k X^l$ with $Z = \sumj \omega^{j} |j\rangle\langle j|$ and $X = \sumj |j\rangle\langle j+1|$, the indices $k$ and $l$ in $\Wkl$ can be identified with phase and shift operations, respectively.
The Bell states of interest for this work are defined by
\begin{align} \label{eq:Omegakl}
    |\Omega_{k,l}\rangle &:= (\Wkl \otimes \id_d) |\Omega_{0,0}\rangle \;,
\end{align}
where $|\Omega_{0,0}\rangle = \frac{1}{\sqrt{d}} \sum_{i=0}^{d-1} |i\rangle\otimes |i\rangle$.
Alternatively, $|\Omega_{k,l}\rangle = \frac{1}{\sqrt{d}} \mathrm{vec}(W_{k,l})$, where the vectorization $\mathrm{vec}(\cdot)$ is defined by its action on the computational basis states of $\linop(\hs)$ by $\mathrm{vec}(|i\rangle\langle j|) = |i\rangle\otimes|j\rangle$ (which extends to general operators $M\in\linop(\hs)$ by linearity).
The set $\{|\Omega_{k,l}\rangle\}_{k,l=0}^{d-1}$ constitutes a complete Bell basis \footnote{This is also called the standard or ``magic'' Bell basis due to its special properties given by the Weyl relations \eqref{eq:Weyl_relations}, and to distinguish it from the generalization considered in Ref.~\cite{popp_special_2024}. This generalization, however, lacks many desirable features of the standard Bell basis, e.g., the generating unitaries generally do not form a group.} because the operators \eqref{eq:Wkl} are orthogonal with respect to the Hilbert-Schmidt inner product, $\mytr(W_{i,j}^\dagger W_{k,l}) = d \, \delta_{i,k} \delta_{j,l}$, and $\langle \Omega_{0,0}| (A\otimes \id) |\Omega_{0,0}\rangle = \frac{1}{d} \mytr(A)$ for all $A\in\linop(\hs_A)$.
The set of states that are diagonal in this basis, also known as Bell-diagonal states, is defined as
\begin{align} \label{eq:magic_simplex}
    \mathcal{M}_d := \left\{ \rho = \sumkl \ckl \, \Pkl \; \Big| \; \sumkl \ckl = 1, \; \ckl \geq 0 \right\} \;,
\end{align}
where $\Pkl := |\Omega_{k,l}\rangle\langle\Omega_{k,l}|$ are the projectors onto the Bell states \eqref{eq:Omegakl}.
Mathematically, equation \eqref{eq:magic_simplex} defines a simplex, which, due to the group structure inherited by the Weyl relations \eqref{eq:Weyl_relations}, is sometimes referred to as a ``magic'' simplex \cite{baumgartner_special_2007}.

%%%%%%%%%%%%%%%%%%%%%%%%%%%%%%%%%%%%%%%%%%%%%%%%%%%%%%%%
\subsection{\texorpdfstring{Group Structure of $\mathcal{M}_d$}{Group Structure of Md}}

Bell-diagonal states $\rho\in\mathcal{M}_d$ possess an underlying group structure\footnote{An example of the following group-theoretic considerations for $d=3$ is presented in App.~\ref{app:group_structure_M_3}.} that is known to have a significant impact on their entanglement class \cite{popp_special_2024, popp_comparing_2023}.
In Sec.~\ref{sec:Realignment_qutrit_BDS} and \ref{sec:PPT_for_qutrit_BDS}, we show that this group property is strongly related to entanglement detection via the PPT and realignment criteria for $\rho\in\mathcal{M}_3$.
To formalize this, we can arrange the $d^2$ coefficients $\ckl$ in \eqref{eq:magic_simplex} in a $d\times d$ coefficient matrix $C = (\ckl)_{kl}$.
For example, for $d=3$, this matrix is given by
\begin{align} \label{eq:weight_matrix_d3}
    C = \begin{pmatrix}
        c_{0,0} & c_{0,1} & c_{0,2}\\
        c_{1,0} & c_{1,1} & c_{1,2} \\
        c_{2,0} & c_{2,1} & c_{2,2}
        \end{pmatrix} \;.
\end{align}
\begin{figure}
    \includegraphics[width=0.85\linewidth]{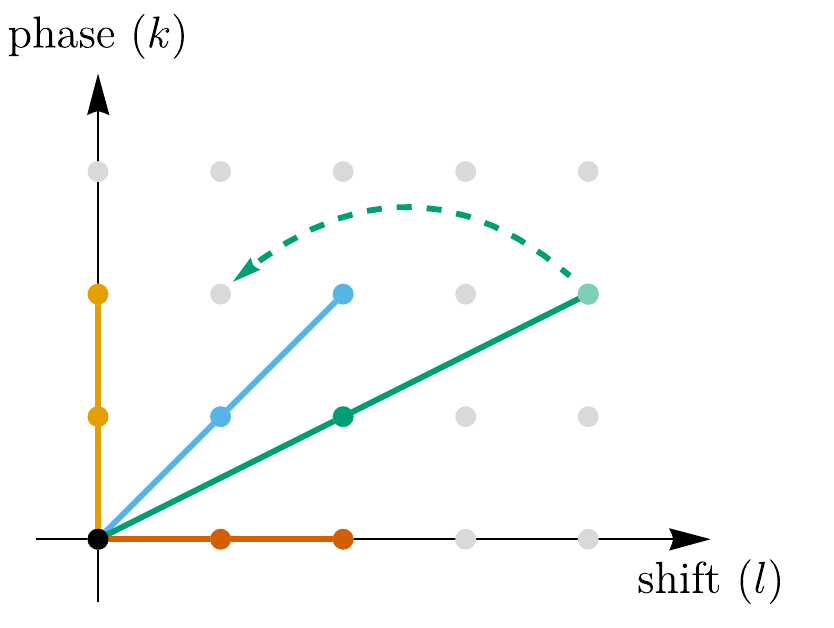}
    \caption{Discrete phase $\ZZ_3^2$ for a Bell-diagonal qutrit state.
    The vertices are associated with the Bell states $\Pkl$.
    The colored lines represent the cosets $S(0,0)$ for all $S\in\mathcal{S}_3$, corresponding to the subgroup states $\rho_\ell$ with $(0,0)\in\ell \in\mathcal{C}_3$.
    For more details, see App.~\ref{app:group_structure_M_3}.
    }
    \label{fig:discrete_phase_space}
\end{figure}
Similarly, we can identify the $d^2$ basis states $\Pkl$ in \eqref{eq:magic_simplex} with points in a discrete phase space $\ZZ_d^2$ (cf. Ref.~\cite{baumgartner_state_2006}; see Fig.~\ref{fig:discrete_phase_space}) where the mapping $\Pkl \mapsto (W_{i,j} \otimes \id_d)\,\Pkl\,(W_{i,j}^\dagger \otimes \id_d) = P_{k+i,l+j}$ is an automorphism.
In particular, $\Pkl = (\Wkl\otimes\id_d) \, P_{0,0} \, (\Wkl^\dagger \otimes \id_d)$, and thus each $\Pkl \in\mathcal{M}_d$ can be uniquely associated to one $\Wkl$.
Admitting a slight abuse of notation (cf. App.~\ref{app:proof_isomorphism}), we define the projective Weyl-Heisenberg group as
\begin{align}\label{eq:W-H-group_projective}
    \mathcal{W}_d := \{\Wkl\}_{k,l\in\ZZ_d} \cong \ZZ_d^2 \;,
\end{align}
where we identify elements that differ only up to a phase (cf. Sec.~4 in Ref.~\cite{baumgartner_special_2007}).
Here, $\ZZ_d^2 = \ZZ_d \times \ZZ_d$ is the direct product of two copies of the cyclic group $\ZZ_d$.
We give a rigorous proof of the isomorphism $\mathcal{W}_d \cong \ZZ_d^2$ in App.~\ref{app:proof_isomorphism}.
Consequently, the structure of $\mathcal{M}_d$ is inherited from the group structure of  $\ZZ_d^2$, which is the central object of the following considerations.

Let $\mathcal{S}_d$ be the set of $d$-element subgroups of $\ZZ_d^2$.
Each $S\in\mathcal{S}_d$ gives rise to a line or sublattice in $\ZZ_d^2$ via the coset $S(i,j)=\{(k,l) + (i,j) \,|\, (k,l) \in S\} $ for some $(i,j)\in\ZZ_d^2$ (cf.~\cite{baumgartner_special_2007, baumgartner_state_2006} and Fig.~\ref{fig:discrete_phase_space}).
For $S\in\mathcal{S}_d$, we define the set of all cosets by
\begin{align}\label{eq:striations}
    \mathcal{C}(S) :=  \{S(i,j) \,|\, (i,j)\in\ZZ_d^2\} \;.
\end{align}
As each point $(i,j)\in\ZZ_d^2$ appears exactly once per $\mathcal{C}(S)$, we call $\mathcal{C}(S)$ a ``striation'' of $\ZZ_d^2$ generated by $S\in\mathcal{S}_d$ \cite{gibbons_discrete_2004} (cf. Fig.~\ref{fig:striations}).
\begin{figure}
    \centering
    \includegraphics[width=0.85\linewidth]{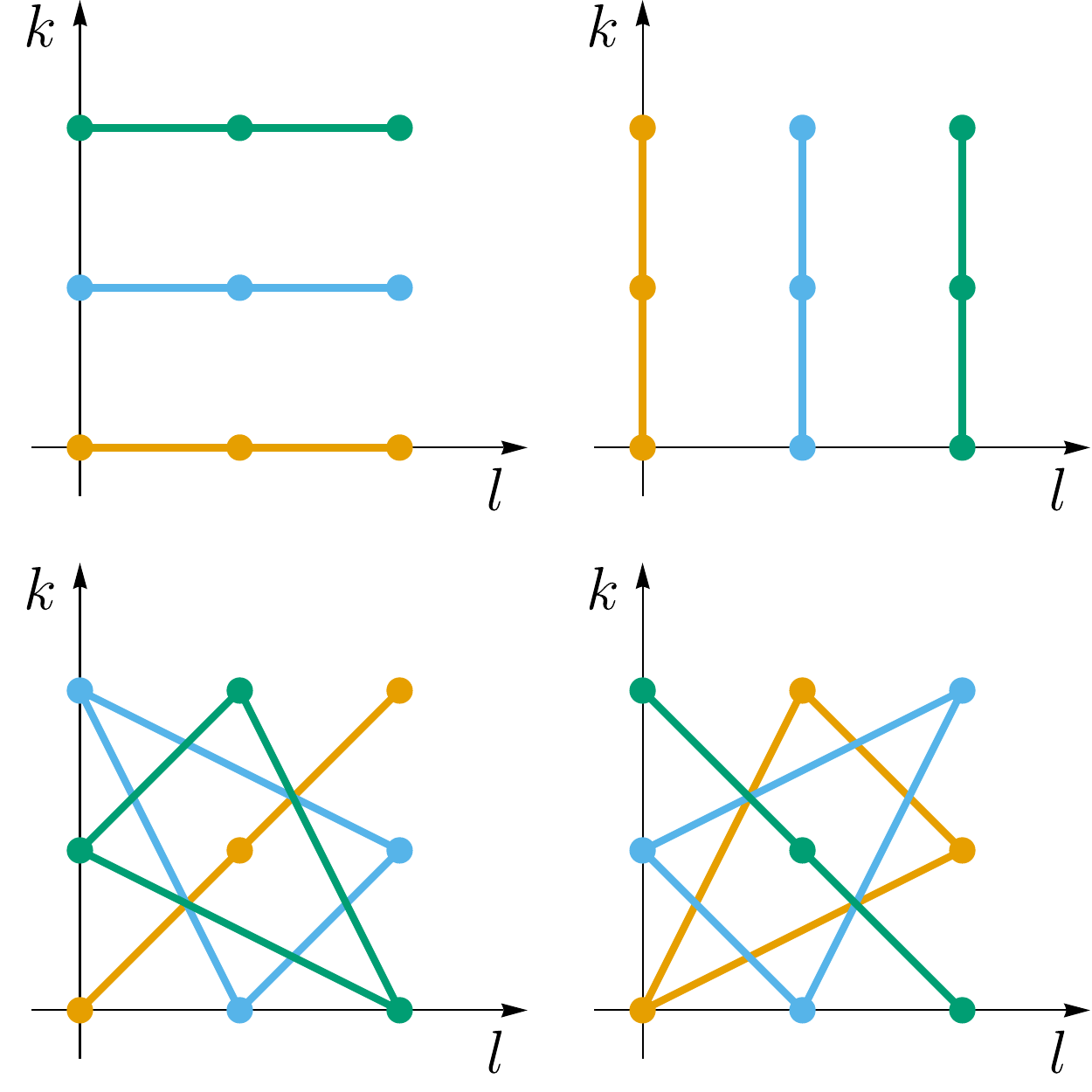}
    \caption{The striations in $\ZZ_3^2$ corresponding to the cosets $\mathcal{C}(S_{0,1})$ (top left), $\mathcal{C}(S_{1,0})$ (top right), $\mathcal{C}(S_{1,1})$ (bottom left), and $\mathcal{C}(S_{1,2})$ (bottom right), given in \eqref{eq:C(S_0,1)}-\eqref{eq:C(S_1,2)} of App.~\ref{app:group_structure_M_3}.
    }
    \label{fig:striations}
\end{figure}
The set of all cosets in $\ZZ_d^2$ is given by
\begin{align} \label{eq:set_all_cosets}
    \mathcal{C}_d := \bigcup_{S\in\mathcal{S}_d} \mathcal{C}(S) \;.
\end{align}
Note that for prime $d$, $\ZZ_d^2$ is an affine plane $\mathrm{AG}(2, d)$ \cite{beth_design_1999} and every nontrivial subgroup $S\in\mathcal{S}_d$ of order $d$ defines a striation via \eqref{eq:striations}.
For composite $d$, the set of size-$d$ subgroups $\mathcal{S}_d$ and the associated striation picture may be incomplete; in this work, we only use the explicit subgroups and coset expansions for $d = 3$.

Bell-diagonal subgroup states $\{\rho_{\ell}\}_{\ell\in\mathcal{C}_d}$ are highly symmetric separable states in $\mathcal{M}_d$ \cite{baumgartner_special_2007}.
For every coset $\ell \in\mathcal{C}_d$, they are defined by having $c_{k,l} = \frac{1}{d}$ for $(k,l)\in\ell$, while all other $\ckl$ are zero, i.e.,
\begin{align} \label{eq:subgroup_state}
    \rho_\ell = \frac{1}{d} \sum_{(k,l) \in\ell} \Pkl \;.
\end{align}
They form the outermost separable states in $\mathcal{M}_d$, i.e., separable states with highest purity in $\mathcal{M}_d \cap \SEP$ (cf. Thm.~3 and Prop.~6 in \cite{baumgartner_special_2007}), and the convex hull of all subgroup states is known as the (separable) kernel polytope.
This is illustrated in Fig.~\ref{fig:simplex_qubits} which shows the state space $\mathcal{M}_2$ for Bell-diagonal qubits, including the maximally entangled Bell states $\{P_{k,l}\}_{k,l=0}^{1}$, the six separable subgroup states $\{\rho_{\ell_i}\}_{i=1}^{6}$, and the kernel polytope.
Note, however, that except for $d=2$, the kernel polytope is not equivalent to the set of separable states: For $d\geq 3$, there exist separable states $\rho_\SEP \in \mathcal{M}_d \cap \SEP$ that are not convex combinations of the subgroup states $\rho_\ell$.
\begin{figure}
    \centering
    \includegraphics[width=0.75\linewidth]{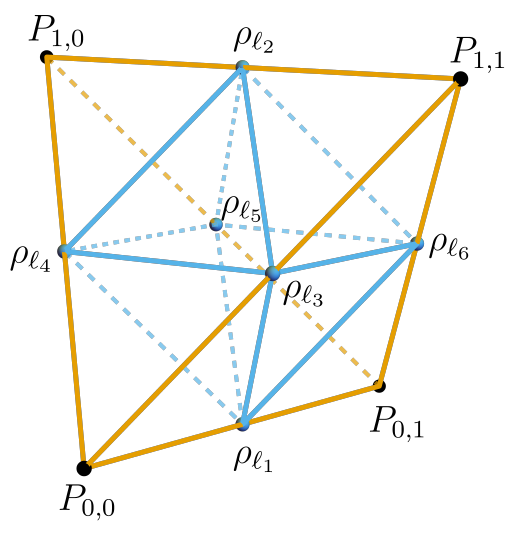}
    \caption{The state space $\mathcal{M}_2$ of Bell-diagonal qubits, also known as ``magic'' simplex (orange tetrahedron).
    The maximally entangled Bell state projectors $P_{k,l}$ are the extreme points of the set, whereas the subgroup states $\rho_{\ell_i}$ are extreme points of the set of separable states (blue octahedron).
    Note that in this case, the separable states coincide with the convex hull of the six subgroup states (kernel polytope), which is generally not true for $d\geq3$ (cf. Ref.~\cite{popp_comparing_2023}).
    }
    \label{fig:simplex_qubits}
\end{figure}

%%%%%%%%%%%%%%%%%%%%%%%%%%%%%%%%%%%%%%%%%%%%%%%%%%%%%%%%
\subsection{\texorpdfstring{Bloch Vector of $\rho\in\mathcal{M}_d$ and the Discrete Fourier Transform}{Bloch Vector of rho in Md and the Discrete Fourier Transform}}
\label{sec:bloch_vec_and_DFT}

As the Weyl-Heisenberg operators \eqref{eq:Wkl} constitute a basis for the space of $d$-dimensional density matrices \cite{bengtsson_geometry_2006}, any bipartite qudit state can be written as
\begin{align} \label{eq:Bloch_repr_general}
    \rho = \frac{1}{d^2} \sum_{i,j,k,l=0}^{d-1} \beta_{ij,kl} \, W_{i,j} \otimes W_{k,l} \;,
\end{align}
where $\beta_{ij,kl}$ are the elements defining the generalized Bloch vector (also known as correlation tensor).
Using $\beta_{ij,kl} = \mytr\left( \rho \, (W_{i,j}\otimes W_{k,l})^\dagger \right)$ and $\rho\in\mathcal{M}_d$, we obtain
\begin{align} \label{eq:Bloch_vector_beta}
    \beta_{ij,kl} = \sum_{x,y=0}^{d-1} c_{x,y} \,\omega^{yi-xj}  \, \delta_{i,-k} \delta_{jl} \;.
\end{align}
With this, we can write \cite{narnhofer_entanglement_2006}
\begin{align} \label{eq:BDS_with_Bloch_vector_A}
    \rho = \frac{1}{d^2} \sum_{i,j=0}^{d-1} B_{j,i} \, W_{i,j} \otimes W_{i,j}^\ast \;,
\end{align}
where we defined the (reduced) Bloch vector 
\begin{align} \label{eq:reduced_bloch_vector_BDS}
    B_{j,i} := \beta_{ij,(-i)j} = \sumkl c_{k,l} \, \omega^{il-jk} \;.
\end{align}
The form of $\rho\in\mathcal{M}_d$ in \eqref{eq:BDS_with_Bloch_vector_A} can equivalently be derived by applying the Weyl twirl channel to the state \eqref{eq:Bloch_repr_general} (see Ref.~\cite{popp_special_2024}).

Note that the Bloch matrix $B := (B_{j,i})_{ji}$ is obtained from the coefficient matrix $C=(\ckl)_{kl}$ by a discrete Fourier transform (DFT) (cf. Ref.~\cite{narnhofer_entanglement_2006}).
In particular, with the unitary DFT matrix
\begin{align} \label{eq:DFT-matrix}
    F:=\frac{1}{\sqrt{d}} \sum_{i,j=0}^{d-1} \omega^{-ij} |i\rangle\langle j| \;,
\end{align}
we get $B = d \, FCF^\dagger$.
Consequently, the Bloch vector $B$ of $\rho\in\mathcal{M}_d$ carries information about the discrete phase space $\ZZ_d^2$ and the Weyl-Heisenberg group $\mathcal{W}_d$: Coset probabilities in $C$ get mapped to their respective annihilator subgroup cosets via the DFT (cf. Thm.~2.7.1 in Ref.~\cite{rudin_fourier_1962}).
In other words, $C$ contains the phase-space probability density function associated with $\rho$ via the coefficients $\{c_{k,l}\}_{k,l=0}^{d-1}$ on $\ZZ_d^2$, whereas $B$ encodes the corresponding discrete Fourier transform on $\ZZ_d^2$ of $C$: for $(i,j)\in\ZZ_d^2$, the entry $B_{j,i}$ in \eqref{eq:reduced_bloch_vector_BDS} is the Fourier coefficient of the function $(k,l)\mapsto c_{k,l}$ with respect to the character $\chi_{i,j}(k,l):=\omega^{il-jk}$.

As preparation for the next section, we conclude here with a short discussion of the properties of $B$ concerning the entanglement class of $\rho\in\mathcal{M}_d$, which is not affected by specific symmetry transformations of $C$ \cite{baumgartner_special_2007}.
For example, for prime dimension $d$, we can understand $\ZZ_d^2$ as the affine plane $\mathrm{AG}(2,d)$ \cite{beth_design_1999}, where these symmetries are given by the affine group $\mathrm{AGL}(2,d)$ (acting as $x \mapsto Ax+b$ with $A\in\mathrm{GL}(2,d)$ and $x,b\in\ZZ_d^2$).
Note, however, that while these transformations preserve entanglement, they cannot be realized by local unitaries acting on the subsystems alone.
Ref.~\cite{baumgartner_special_2007} shows that local unitaries (realizing shifts, vertical shears, and rotations in $\ZZ_d^2$) as well as a global anti-unitary (inverting the phase $k$ by conjugation of $\rho$) are required to realize all symmetries.
However, as $B$ generally changes under these mappings, the information about the entanglement class contained in $B$ must be invariant under such transformations as well.

%%%%%%%%%%%%%%%%%%%%%%%%%%%%%%%%%%%%%%%%%%%%%%%%%%%%%%%%%%%%%%%%%%%%%%%%%%%%%%%%%%%%%%%%%%%%%
\section{Realignment Criterion for Bell-Diagonal States} \label{sec:Realignment_for_BDS}

The realignment criterion (Thm.~\ref{thm:realignment}) can be used to detect entangled quantum states, but it lacks a clear operational and geometric meaning.
In this section, we present the realignment criterion for Bell-diagonal qudit-qudit states $\rho\in\mathcal{M}_d$ in a form that admits a clear geometric interpretation in terms of the Bloch vector:
\begin{mytheorem}
\label{thm:realignment_BDS}
    The realignment criterion for $\rho\in\mathcal{M}_d$ is given by
    \begin{align} \label{eq:realignment_bds}
        \Vert B \Vert_1 >d \Longrightarrow \rho \in \ENT \;,
    \end{align}
    where $B=(B_{j,i})_{ji}$ is the Bloch vector of $\rho$ as defined in \eqref{eq:reduced_bloch_vector_BDS}, and $\Vert B\Vert_1 = \sum_{i,j=0}^{d-1} \left| B_{j,i} \right|$ denotes the entry-wise 1-norm.
\end{mytheorem}
Before we prove the theorem, two remarks are in order.
First, this theorem reduces the computational complexity of computing the realignment criterion for $\rho\in\mathcal{M}_d$ compared to Thm.~\ref{thm:realignment}.
Particularly, as $B$ is proportional to the DFT of $C$, evaluating \eqref{eq:realignment_bds} scales only as $\mathcal{O}(d^2\log(d))$ instead of $\mathcal{O}(d^6)$ for \eqref{eq:realignment_crit} (cf. Ref.~\cite{golub_matrix_2013}).
Second, \eqref{eq:realignment_bds} is equivalent to \eqref{eq:realignment_correlation_tensor} if the correlation tensor is diagonal (as is the case for Bell-diagonal states; cf. \eqref{eq:BDS_with_Bloch_vector_A}), in which case its trace norm reduces to the entry-wise 1-norm.
The extra factor $d$ in \eqref{eq:realignment_bds} is due to the unnormalized Weyl matrices satisfying $\mytr(W_{i,j}W_{k,l}^\dagger) = d \, \delta_{i,k}\delta_{j,l}$.
Also note that a similar result for Pauli-diagonal matrices is presented in Ref.~\cite{muller-hermes_decomposable_2021}.

\begin{proof}[Proof of Thm.~\ref{thm:realignment_BDS}]
    Using \eqref{eq:Wkl} and \eqref{eq:Omegakl}, we write $\rho\in\mathcal{M}_d$ as
    \begin{align} \label{eq:rho_expanded}
        \rho =  \frac{1}{d} \sum_{i,j,k,l=0}^{d-1} c_{k,l} \,\omega^{k(i-j)} |i-l\rangle\langle j-l| \otimes |i\rangle\langle j| \;.
    \end{align}
    With \eqref{eq:realignment_operator}, the realigned operator $\rho_R := \hat{R}(\rho)$ evaluates to
    \begin{align}
        \rho_R &= \frac{1}{d} \sum_{i,j,k,l=0}^{d-1} \ckl \,\omega^{k(i-j)} \, |i-l\rangle\langle i| \otimes |j-l\rangle\langle j| \\
        &= \frac{1}{d} \sumkl \ckl \, W_{k,l}\otimes W_{k,l}^\ast \;.
    \end{align}
    The eigenvectors of $W_{k,l}\otimes W_{k,l}^\ast$ are the Bell states $|\Omega_{i,j}\rangle$ in \eqref{eq:Omegakl} with respective eigenvalue $\omega^{il-jk}$ \cite{popp_special_2024}.
    Hence, we obtain
    \begin{align}
        \rho_R \, |\Omega_{i,j}\rangle
        =  \frac{1}{d} \sumkl \ckl \, \omega^{il-jk} |\Omega_{i,j}\rangle \;,
    \end{align}
    and we can diagonalize $\rho_R$ via the unitary $U=\sumkl (|k\rangle \otimes |l\rangle )\langle \Okl|$ to obtain $U\rho_R U^\dagger = \frac{1}{d} \sum_{i,j=0}^{d-1} B_{j,i} |i\rangle\langle i| \otimes |j\rangle\langle j|$.
    Thus,
    \begin{align}
        \mytr\left( \sqrt{\rho_R^\dagger \rho_R} \right)
        =  \frac{1}{d} \sum_{i,j=0}^{d-1} |B_{j,i}| \;.
    \end{align}
    Inserting this into Thm.~\ref{thm:realignment} completes the proof.
\end{proof}

%%%%%%%%%%%%%%%%%%%%%%%%%%%%%%%%%%%%%%%%%%%%%%%%%%%%%%%%%%%%%%%%%%%%%%%%%%%%%%%%%%%%%%%%%%%%%
\subsection{\texorpdfstring{Realignment Criterion for Bell-Diagonal Qutrits ($d=3$)}{Realignment Criterion for Bell-Diagonal Qutrits (d3)}} \label{sec:Realignment_qutrit_BDS}

Thm.~\ref{thm:realignment_BDS} gives a geometric entanglement condition based on the Bloch vector of $\rho\in\mathcal{M}_d$.
Because the Bloch vector $B$ in \eqref{eq:reduced_bloch_vector_BDS} can be obtained by the DFT of the coefficient matrix $C$, it encodes the structure of the discrete phase space $\ZZ_d^2$.
In particular, $B$ contains information about the subgroups of $\ZZ_d^2$.

To exemplify this, we present the explicit version of Thm.~\ref{thm:realignment_BDS} for Bell-diagonal qutrit states $\rho\in\mathcal{M}_3$.
We begin by using $|B_{0,0}| = |\sumkl c_{k,l}| = 1$ to obtain
\begin{align}
    \Vert B \Vert_1 &= 1+\sum_{\substack{i,j=0 \\ (i,j)\neq (0,0)}}^{2} \left| B_{j,i} \right| \;.
\end{align}
Second, we compute
\begin{align}
    |B_{0,1}|^2 &= \sum_{k,l,r,s=0}^{2} \, c_{k,l} \, c_{r,s} \, \omega^{l-s} \\
    &= \sum_{k,l,r,s=0}^{2} \, c_{k,l} \, c_{r,s} \, \omega^{-l+s} \\
    &= \sum_{k,l,r,s=0}^{2} \, c_{k,l} \, c_{r,s} \, \omega^{2l-2s} = |B_{0,2}|^2 \;,
\end{align}
where we used that $|B_{0,1}|\in\R$ and $c_{k,l} \, c_{r,s}\in\R$ in the second line, and $\omega^{2} = \omega^{-1}$ for $d=3$ in the third line.
Similarly, we obtain $|B_{1,0}| = |B_{2,0}|$, $|B_{1,2}| = |B_{2,1}|$, and $|B_{1,1}| = |B_{2,2}|$.
The inequality in \eqref{eq:realignment_bds} thus becomes
\begin{align} \label{eq:geometric_realignment_inequality}
    2 (|B_{0,1}|+|B_{1,0}|+|B_{1,1}|+|B_{1,2}|) >2 \;,
\end{align}
which, when inserted into \eqref{eq:realignment_bds}, serves as a geometric version of the realignment criterion for $\rho\in\mathcal{M}_3$, i.e., all Bell-diagonal qutrit states with Bloch vector $(B_{j,i})_{ji}$ satisfying the inequality \eqref{eq:geometric_realignment_inequality} are entangled.
Simplifying the terms on the left hand side using $\sum_{k,l=0}^{2}c_{k,l} = 1$, the subgroups $S\in\mathcal{S}_3$, and the corresponding cosets $\mathcal{C}(S)$ presented in App.~\ref{app:group_structure_M_3}, the realignment criterion \eqref{eq:realignment_bds} for $d=3$ can be written in terms of the coset probabilities $\sum_{(i,j)\in\ell}c_{i,j}$ with $\ell \in \mathcal{C}(S)$ as
\begin{align}\label{eq:realignment_BDS_qutrits}
    \sum_{S\in\mathcal{S}_3} \sqrt{6\sum_{\ell \in \mathcal{C}(S)} \left( \sum_{(i,j) \in \ell} c_{i,j}  \right)^2 -2} \; >2 \;\Rightarrow\;\rho\in\ENT \;.
\end{align}
In words: First, $\sum_{S\in\mathcal{S}_3}$ sums over all subgroups $S\in\mathcal{S}_3$ generating horizontal, vertical, diagonal, and anti-diagonal striations $\mathcal{C}(S)$; second, $\sum_{\ell\in\mathcal{C}(S)}$ sums over all cosets $\ell$ contained in the striation $\mathcal{C}(S)$; third, $\sum_{(i,j)\in\ell}$ sums the coefficients $c_{i,j}$ of each coset $\ell\in\mathcal{C}(S)$.
Consequently, the subgroup structure of $\ZZ_3^2$ contained in the Bloch vector \eqref{eq:reduced_bloch_vector_BDS} is crucial for entanglement detection in $\mathcal{M}_3$ using the realignment criterion.

Furthermore, \eqref{eq:realignment_BDS_qutrits} entails a measurement procedure to obtain the realignment value of $\rho\in\mathcal{M}_3$ by measuring the coset probabilities $\sum_{(i,j)\in\ell}c_{i,j}$ for all $\ell \in \mathcal{C}(S)$.
This requires projecting the state for each $S\in\mathcal{S}_3$ onto the different coset subspaces using the projection operators in $\Pi(S) = \{\sum_{(i,j)\in\ell} P_{i,j} \, | \, \ell \in \mathcal{C}(S)\}$.
In this way, each subgroup gives rise to a projective measurement with measurement probabilities $\{\sum_{(i,j)\in\ell}c_{i,j}\}_{\ell \in \mathcal{C}(S)}$, corresponding to the striation $\mathcal{C}(S)$.
In App.~\ref{app:group_structure_M_3}, we sketch how this measurement can be conducted using local projective measurements.
This requires four different measurement settings, one for each striation $\mathcal{C}(S)$.
Consequently, one can determine whether a state $\rho\in\mathcal{M}_3$ is entangled according to the realignment criterion using only four different local measurements.

%%%%%%%%%%%%%%%%%%%%%%%%%%%%%%%%%%%%%%%%%%%%%%%%%%%%%%%%%%%%%%%%%%%%%%%%%%%%%%%%%%%%%%%%%%%%%
\section{\texorpdfstring{PPT Criterion for Bell-Diagonal Qutrits ($d=3$)}{PPT Criterion for Bell-Diagonal Qutrits (d3)}} \label{sec:PPT_for_qutrit_BDS}

For bipartite qubit states $\rho\in\denop(\C^2\otimes\C^2)$, the PPT criterion (Thm.~\ref{thm:PPT_crit}) is necessary and sufficient for entanglement \cite{horodecki_separability_1996}, and it even holds that \cite{augusiak_universal_2008}
\begin{align}\label{eq:equivalences_ENT_NPT_2}
    \rho \in \ENT \Longleftrightarrow \rho^\Gamma \ngeq 0 \Longleftrightarrow \det(\rho^\Gamma) <0 \;.
\end{align}
For general qudit-qudit states $\rho\in\denop(\C^{d_A}\otimes\C^{d_B})$ with $d_A d_B>6$, these equivalences fail and we get the weaker statement
\begin{align} \label{eq:equivalences_ENT_NPT_d}
    \rho \in \ENT \Longleftarrow \rho^\Gamma \ngeq 0 \Longleftarrow \det(\rho^\Gamma) <0 \;.
\end{align}
In this section, we show that for Bell-diagonal qutrit-qutrit states $\rho\in\mathcal{M}_3$, the equivalence for the second implication in \eqref{eq:equivalences_ENT_NPT_d} is restored, namely, $\rho^\Gamma \ngeq 0 \Longleftrightarrow \det(\rho^\Gamma) <0$.
In contrast, even for Bell-diagonal qutrit states, the first implication cannot be promoted to an equivalence because of the existence of PPT-entangled states in $\mathcal{M}_3$, i.e., $\ENT \cap \PPT\cap\mathcal{M}_3 \neq \varnothing$.
Furthermore, our result underscores that the subgroup structure of $\ZZ_3^2 $ plays a crucial role in distinguishing between $\PPT \cap\mathcal{M}_3$ and $\NPT \cap\mathcal{M}_3$.
\begin{mytheorem} \label{thm:PPT_crit_BDS}
    Let $\rho\in\mathcal{M}_3$ and let $\rho^\Gamma$ be the partial transpose of $\rho$.
    It holds that
    \begin{align} \label{eq:PPT_crit_BDS}
        \rho^\Gamma \ngeq 0 \Longleftrightarrow \det(\rho^\Gamma) <0 \;.
    \end{align}
    Explicitly, we have
    \begin{align} \label{eq:det(A_0)}
        \rho^\Gamma \ngeq 0 \; \Longleftrightarrow \;3\sum_{\ell\in\mathcal{C}_3}\left( \prod_{(i,j)\in \ell} c_{i,j} \right) < \sum_{k,l=0}^{2} c_{k,l}^3 \;.
    \end{align}
\end{mytheorem}
The idea of the proof is to reduce the spectrum of the partial transposition of $\rho\in\mathcal{M}_3$ to the spectrum of a $3\times 3$ matrix $A_0$ with only one non-contradicting assignment of the signs of its eigenvalues if $\rho^\Gamma \ngeq 0$.
\begin{proof}
We use \eqref{eq:rho_expanded} to compute the partial transposition $\rho^\Gamma := (\id \otimes T)(\rho)$ and apply the basis transformation $U=\sum_{i,j=0}^{2} (|i\rangle\otimes |j\rangle)\langle\Omega_{i,j}|$, finding that $\rho^\Gamma$ is block-diagonal in the Bell basis, i.e.,
\begin{align} \label{eq:rho_gamma_block_diagonal}
    U\rho^\Gamma U^\dagger = \sum_{m=0}^{2} |m\rangle\langle m| \otimes A_m \;,
\end{align}
where we defined
\begin{align} \label{eq:A_m}
    A_m := \frac{1}{9} \sum_{i,j,k,l=0}^{2} \omega^{j(m-k)-i(l+j)} \ckl W_{i,j} \;.
\end{align}
The hermitian blocks $A_m$ are known to be unitarily equivalent for odd dimension $d$ \cite{baumgartner_special_2007}, and we have $A_m = W_{m,0}^\dagger A_0 W_{m,0}$ for $m\in\{0,1,2\}$.
Hence, each $A_m$ has the same eigenvalues and thus the spectrum of $A_0$ completely determines the spectrum of $\rho^\Gamma$.
Also, the multiplicity of each eigenvalue of $\rho^\Gamma$ must be divisible by 3.
Furthermore, $\mytr(A_0) = \frac{1}{3} \sumkl \ckl = \frac{1}{3}$ and, because $\det(\rho^\Gamma) = \prod_{m=0}^{2} \det\left(A_m\right) = \det(A_0)^3$, we have
\begin{align} \label{eq:det(rhoGamma)<det(A0)}
    \det(\rho^\Gamma)<0 \Longleftrightarrow \det(A_0) <0 \;.
\end{align}

The spectrum of $A_0$ is restricted as follows:
Let us denote the minimal and maximal eigenvalues of a hermitian matrix $M$ by $\lambda_\mathrm{min}(M)$ and $\lambda_\mathrm{max}(M)$, respectively.
The qutrit Bell basis projectors $\Pkl = |\Omega_{k,l}\rangle\langle\Omega_{k,l}| \in\mathcal{M}_3$ satisfy $\lambda_\mathrm{min}(\Pkl^\Gamma) = -\frac{1}{3}$ and $\lambda_\mathrm{max}(\Pkl^\Gamma) = \frac{1}{3}$ for all $k,l\in\{0,1,2\}$.
Applying Weyl's inequalities \cite{bhatia_matrix_1997} to $\rho^\Gamma$ for $\rho\in\mathcal{M}_3$ yields
\begin{align}
    \lambda_\mathrm{min}(\rho^\Gamma) &= \lambda_\mathrm{min}\left(\sum_{k,l=0}^{2} \ckl \Pkl^\Gamma \right) \\
    &\geq \sum_{k,l=0}^{2} \ckl \, \lambda_\mathrm{min}\left( \Pkl^\Gamma \right) = - \frac{1}{3} \;,
\end{align}
and similarly $\lambda_\mathrm{max}(\rho^\Gamma) \leq \frac{1}{3}$, i.e., the spectrum of $\rho^\Gamma$ and $A_0$ is contained in $[-\frac{1}{3}, \frac{1}{3}]$.

Now, let us show that only $\det(A_0)<0$ (and thus $\det(\rho^\Gamma)<0$) is consistent with $\rho^\Gamma\ngeq0$.
The partial transpose of any $\rho\in\denop(\hs_A \otimes \hs_B)$ with $\dim(\hs_A) = d_A$ and $\dim(\hs_B) = d_B$ has at most $n_- =(d_A-1)(d_B-1)$ negative eigenvalues \cite{rana_negative_2013, johnston_non-positive-partial-transpose_2013}\footnote{Note that this bound on the number of negative eigenvalues also holds more generally for entanglement witnesses \cite{Sarbicki_Spectralpropertiesentanglement_2008}, reflecting the block-positivity of these operators.}.
For $\rho\in\mathcal{M}_3$ we have $d_A = d_B = 3$ and thus $n_- \leq 4$.
By the multiplicity restriction on the eigenvalues of $\rho^\Gamma$, we conclude that either $n_-=0$ or $n_- = 3$.
The case $n_-=0$ is equivalent to $\rho^\Gamma \geq0$ and $\det(\rho^\Gamma) \geq0$.
For $n_-=3$ we have $\rho^\Gamma \ngeq 0$, and $A_0$ has exactly one negative eigenvalue $\lambda_1 <0$.
This leaves three distinct cases for the other two non-negative eigenvalues $\lambda_2$ and $\lambda_3$ of $A_0$.
First, if $\lambda_2 = \lambda_3 = 0$, we obtain $\mytr(A_0) = \lambda_1 <0$, which contradicts $\mytr(A_0) = \frac{1}{3}$.
Second, if $\lambda_2 >\lambda_3 =0$ we have $ \frac{1}{3} = \mytr(A_0) = \lambda_1 + \lambda_2$, and thus $\lambda_2 = \frac{1}{3} - \lambda_1 >\frac{1}{3}$.
But this contradicts that the spectrum of $A_0$ is contained in $[-\frac{1}{3},\frac{1}{3}]$.
The third and only consistent case for $\rho^\Gamma\ngeq0$ is $\lambda_1 <0<\lambda_2, \lambda_3$ with $\det(A_0) = \lambda_1 \lambda_2\lambda_3<0$.
Thus, we have shown that $\rho^\Gamma \ngeq 0 \Longrightarrow \det(A_0) <0$.
Together with \eqref{eq:det(rhoGamma)<det(A0)} and the (trivial) implication $\det(\rho^\Gamma)<0 \Longrightarrow \rho^\Gamma \ngeq 0$, we get the equivalence \eqref{eq:PPT_crit_BDS}.
Finally, the explicit formula \eqref{eq:det(A_0)} is obtained by computing $\det(A_0)$ using \eqref{eq:A_m}.
\end{proof}

Thm.~\ref{thm:PPT_crit_BDS} has interesting implications regarding the geometric properties of the sets $\PPT\cap\mathcal{M}_3$ and $\NPT\cap\mathcal{M}_3$.
To see this, let us define the boundary $\partial(\PPT)$ of PPT as states $\rho\in\PPT\cap\mathcal{M}_3$ such that every $\epsilon$-neighborhood around $\rho$ is neither contained in PPT nor NPT.
As $\det(\rho^\Gamma)$ is a continuous function of the coefficients $\ckl$, boundary states $\rho\in\partial(\PPT)$ satisfy $\det(\rho^\Gamma) =0$.
Using \eqref{eq:det(A_0)}, this is equivalent to
\begin{align} \label{eq:PPT_bnd_condition}
    \rho \in \partial (\PPT) \; \Longrightarrow \;3\sum_{\ell\in\mathcal{C}_3}\left( \prod_{(i,j)\in \ell} c_{i,j} \right) = \sum_{k,l=0}^{2} c_{k,l}^3 \;.
\end{align}
Simple examples include the separable subgroup states $\{\rho_\ell\}_{\ell\in\mathcal{C}_3}$ (cf. Sec.~\ref{sec:BDS}).
We can reformulate this as a dual characterization of the set $\PPT\cap\mathcal{M}_3$ by hyperplanes, defined by
\begin{gather} 
    W_\NPT=\sum_{i,j=0}^{2} \kappa_{i,j} P_{i,j} \;, \label{eq:W_NPT}\\
    \kappa_{i,j} = -c_{i,j}^2 +
    \sum_{\{(i,j), (k,l), (m,n)\}\in\mathcal{C}_3}
    c_{k,l} \, c_{m,n}\;. \label{eq:kappa_ij}
\end{gather}
The explicit form of $\kappa_{i,j}$ is presented in App.~\ref{app:explicit_form_of_W_NPT}.
If the coefficients $\ckl$ in $\rho$ and $\kappa_{i,j}$ coincide, a short calculation yields:
\begin{align}
    \mytr(\rho W_\NPT) &= \mytr\Bigg[ \Bigg( \sum_{k,l=0}^{2} c_{k,l} P_{k,l} \Bigg)
    \Bigg( \sum_{i,j=0}^{2} \kappa_{i,j} P_{i,j} \Bigg) \Bigg] \\
    &= \sum_{i,j=0}^{2} c_{i,j} \kappa_{i,j}\\
    &= 3\sum_{\ell\in\mathcal{C}_3}\left( \prod_{(i,j)\in \ell} c_{i,j} \right) - \sum_{k,l=0}^{2} c_{k,l}^3 \;,
\end{align}
where we used that $\mytr(P_{k,l} P_{i,j}) = \delta_{k,i} \delta_{l,j}$.
Hence, in this case, it holds by \eqref{eq:PPT_crit_BDS} and \eqref{eq:PPT_bnd_condition} that $\rho^\Gamma \ngeq 0 \Longleftrightarrow \mytr(\rho \, W_\NPT) <0$ and $\rho\in\partial(\PPT) \Longrightarrow \mytr(\rho \, W_\NPT) = 0$.
We note, however, that the observable $W_\NPT$ is not necessarily an entanglement witness as neither $\mytr(\sigma_\mathrm{s} \,W_\NPT) \geq0$ for all $\sigma_\mathrm{s}\in\SEP$, nor $\mytr(\sigma_\mathrm{e} \, W_\NPT) <0$ for some $\sigma_\mathrm{e} \in \ENT$ are guaranteed.
Indeed, $W_\NPT$ can be the zero matrix if all $\kappa_{i,j}$ vanish, as is the case for all subgroup states $\{\rho_\ell\}_{\ell\in\mathcal{C}_3}$; or it can become a positive operator, as is the case for the maximally mixed state with $\ckl = \frac{1}{9}$ for all $k,l\in\{0,1,2\}$.
Furthermore, the Weyl-covariant \cite{watrous_theory_2018} map $\Lambda_\NPT$ corresponding to $W_\NPT$ via the Choi-Jamio\l{}kowski isomorphism \cite{choi_completely_1975, jamiolkowski_linear_1972} is given by
\begin{gather}
    \begin{aligned}
        \Lambda_\NPT (\rho) &= 3 \mytr_B\left(W_\NPT \left(\id_A \otimes \rho^T\right)\right) \\
        &= \sum_{i,j=0}^{2} \kappa_{i,j} W_{i,j} \, \rho \, W_{i,j}^\dagger \;,
    \end{aligned}
\end{gather}
where $\mytr_B(\cdot)$ denotes the partial trace over the second subsystem.
It is shown in Ref.~\cite{baumgartner_state_2006, baumgartner_special_2007} that $W_\NPT$ satisfies $\mytr(\sigma_\mathrm{s} \, W_\NPT) \geq 0$ for all $\sigma_\mathrm{s} \in \SEP$ if and only if $\Lambda_\NPT$ is a positive map; if additionally at least one $\kappa_{i,j}$ is negative, then $W_\NPT$ constitutes an entanglement witness for some $\sigma_\mathrm{e}\in\ENT$, i.e., $\mytr(\sigma_\mathrm{e} \, W_\NPT) <0$.
One example can be constructed from the isotropic state $\rho_\mathrm{iso}(\beta) = d^{-2} (1-\beta) \id_{d^2} + \beta \, P_{0,0} \in \mathcal{M}_d$ with $-\frac{1}{d^2-1} \leq \beta \leq 1$, which is NPT for $\frac{1}{d+1}<\beta\leq 1$ \cite{horodecki_quantum_2009}.
Hence, for $d=3$, we have $\rho_\mathrm{iso}(1/4)\in\partial(\PPT)$.
Using \eqref{eq:W_NPT} and \eqref{eq:kappa_ij}, we obtain $W^{(\mathrm{iso})}_\NPT = \frac{1}{24} \id_{d^2} - \frac{1}{8} P_{0,0}$.
To check $\mytr(\sigma_s W^{(\mathrm{iso})}_\NPT)\geq 0$ for all $\sigma_s\in\SEP$, it is enough to consider pure product states $|a,b\rangle :=|a\rangle\otimes |b\rangle \in\C^3\otimes \C^3$ and verify $\langle a,b|W^{(\mathrm{iso})}_\NPT |a,b\rangle \geq 0$, since SEP is the convex hull of pure product states \cite{watrous_theory_2018}, and the trace is linear.
First, we define $a_i = \langle i|a\rangle$ and $b_i = \langle i|b\rangle$ satisfying the normalization $\sum_{i=0}^{2} |a_i|^2 = \sum_{i=0}^{2} |b_i|^2 = 1$, and compute
\begin{align}
    \langle a,b| P_{0,0} |a,b\rangle &= |\langle \Omega_{0,0} | a,b\rangle |^2 \\
    &= \frac{1}{3} \left| \sum_{i=0}^{2} a_i b_i \right|^2 \\
    &\leq \frac{1}{3} \left( \sum_{i=0}^{2} |a_i|^2 \right) \left( \sum_{i=0}^{2} |b_i|^2 \right) \\
    &= \frac{1}{3} \;,
\end{align}
where in the third line we used the Cauchy-Schwarz inequality.
With this, we obtain
\begin{align}
    \langle a,b|W^{(\mathrm{iso})}_\NPT |a,b\rangle = \frac{1}{24} - \frac{1}{8} \langle a,b| P_{0,0} |a,b\rangle \geq 0 \;,
\end{align}
and thus $\mytr(\sigma_s W^{(\mathrm{iso})}_\NPT)\geq 0$ for all $\sigma_s\in\SEP$.
We also have $\mytr(P_{0,0} W^{(\mathrm{iso})}_\NPT) = -1/12<0$, and thus $|\Omega_{0,0}\rangle$ is detected as entangled by $W^{(\mathrm{iso})}_\NPT$.
Consequently, $W^{(\mathrm{iso})}_\NPT$ is an entanglement witness satisfying $\mytr(\rho_\mathrm{iso}(1/4) W^{(\mathrm{iso})}_\NPT) = 0$.

Lastly, Thm.~\ref{thm:PPT_crit_BDS} has another consequence regarding PPT-entangled states in $\mathcal{M}_3$.
Using the fact that subgroup states are separable \cite{baumgartner_special_2007}, we arrive at the following statement
(the proof is presented in App.~\ref{app:proof_thm_lines_BE}):
\begin{myprop} \label{thm:no_PPT-ENT_with_one_line_zeros}
    There is no PPT-entangled $\rho\in\mathcal{M}_3$ with coefficients $c_{k,l}$ corresponding to at least one coset $\ell\in\mathcal{C}_3$ equal to zero.
\end{myprop}
This serves as a good guiding principle for finding PPT-entangled states in $\mathcal{M}_3$.
For example, the Bell-diagonal states that are defined by the coefficient matrix
\begin{align}
    C = \begin{pmatrix}
        \frac{1}{3}-x & \frac{1}{3}-x & \frac{1}{3}-x\\
        2x & x & 0 \\
        0 & 0 & 0
        \end{pmatrix} \;,
\end{align}
with $x\in[0,1/3]$, satisfy the premise of the proposition and are separable for $x=0$ (reducing to case 3b in App.~\ref{app:proof_thm_lines_BE}) and NPT-entangled for $x\in(0,1/3]$ (cf. case 5 in App.~\ref{app:proof_thm_lines_BE}).
On the contrary, Ref.~\cite{caban_bound_2023} considered allegedly similar-looking states that are (up to the symmetries of $\mathcal{M}_3$; cf. Sec.~\ref{sec:bloch_vec_and_DFT}) equivalent to 
\begin{align}
    C = \begin{pmatrix}
        \frac{1}{3}-x & \frac{1}{3}-x & \frac{1}{3}-x\\
        2x & 0 & 0 \\
        x & 0 & 0
        \end{pmatrix} \;.
\end{align}
These states do not satisfy the requirement of the proposition as no coset $\ell\in\mathcal{C}_3$ has coefficients $\ckl$ equal to zero.
Using \eqref{eq:realignment_BDS_qutrits} and \eqref{eq:det(A_0)}, it can be checked that this state is PPT-entangled for $x\in(0,\frac{2}{15}]$.

%%%%%%%%%%%%%%%%%%%%%%%%%%%%%%%%%%%%%%%%%%%%%%%%%%%%%%%%%%%%%%%%%%%%%%%%%%%%%%%%%%%%%%%%%%%%%
\section{Conclusion and Outlook} \label{sec:conclusion}

The goal of this work was to answer the question of how the group structure of Bell-diagonal states relates to the PPT and realignment criteria.
To this end, we leveraged the symmetry properties of this state family to provide a unified approach connecting these criteria to the group-theoretic features of the system.

Thm.~\ref{thm:realignment_BDS} shows that the realignment criterion for Bell-diagonal states is strongly linked to the geometry of the Bloch ball.
As the Bloch vector of Bell-diagonal states is obtained by the discrete Fourier transform of their defining coefficient matrix, it contains information about the subgroups of $\ZZ_d^2$.
For the specific case of Bell-diagonal qutrits, we find that in the 8-dimensional coefficient space of $\ckl$, the inequality for the realignment criterion is at most quadratic in $\ckl$.
In contrast, the PPT set is defined by a cubic inequality by Thm.~\ref{thm:PPT_crit_BDS}.
Additionally, this explicit version of the PPT criterion for Bell-diagonal qutrit states $\rho\in\mathcal{M}_3$ unveils that the equivalence $\det(\rho^\Gamma) <0 \Longleftrightarrow \rho^\Gamma \ngeq 0$ holds, which is generally not valid for all states of the same dimension.
Consequently, our results not only demonstrate that the group structure of Bell-diagonal states simplifies the analysis of two specific entanglement detection criteria but also provide deeper insights into the mathematical and physical properties of this state family.
In particular, our treatment entails specific measurement procedures to evaluate these criteria experimentally.
Furthermore, the group-theoretic methods we developed enable a methodological analysis of group properties and entanglement structure, which has recently been successfully applied to stabilizer-based entanglement distillation \cite{popp_novel_2025}.

While our analysis focuses on Bell-diagonal states and is limited to $d=3$ for the PPT criterion, adapting this framework to other dimensions or more general classes of quantum states (e.g., Bell-diagonal states of unequal local dimension \cite{moerland_bound_2024}) remains an open challenge.
Extending our treatment of the PPT criterion for Bell-diagonal qudits for $d>3$ is not straightforward, as the spectral properties of $U \, \rho^\Gamma \, U^\dagger$ do not satisfy the requirements used in the proof of Thm.~\ref{thm:PPT_crit_BDS}.

Going forward, it is furthermore an open problem whether for general bipartite qudit states, an entanglement criterion based on the entry-wise 1-norm of the Bloch vector, similar to Thm.~\ref{thm:realignment_BDS}, can be derived.
Ref.~\cite{pittenger_convexity_2002, chruscinski_generalized_2008} present a sufficient criterion for separability, reading $\Vert \beta\Vert_1 \leq 2 \Longrightarrow \rho\in\SEP$ for $\beta$ as in \eqref{eq:Bloch_repr_general}, but to the authors' knowledge, no corresponding sufficient entanglement criterion exists.
A naive approximation using $\Vert \beta \Vert_1 \leq d^2 \, \Vert \beta \Vert_2$ for $\beta\in\C^{d^4}$ yields that $\rho$ is entangled if $\Vert\beta\Vert_1 >d^3$.
However, the maximal value for pure states $|\psi\rangle\in\C^d\otimes \C^d$ is also $\Vert\beta\Vert_1 =d^3$.
A quick numerical optimization for $d=3$ yields $\Vert\beta\Vert_1 \lesssim 25.9735$ for general pure states $|\psi\rangle\in\C^3\otimes\C^3$, whereas separable pure states seem to satisfy $\Vert \beta\Vert_1 \leq 25$.
The discrepancy of $25.9735<27=3^3$ stems from the fact that equality in $\Vert \beta \Vert_1 \leq d^2 \, \Vert \beta \Vert_2$ holds if and only if $\beta$ and the vector $(1,\dots,1)^T$ are linearly dependent.
However, this cannot be as $\frac{1}{d^2}\sum_{i,j,k,l=0}^{d-1} W_{i,j}\otimes W_{k,l}$ does not represent a valid quantum state.

To conclude, by bridging the gap between group theory and entanglement detection, this work provides a novel perspective on the structure of quantum states, specifically for qutrit-qutrit systems, paving the way for further advancements in quantum information science.

%%%%%%%%%%%%%%%%%%%%%%%%%%%%%%%%%%%%%%%%%%%%%%%%%%%%%%%%%%%%%%%%%%%%%%%%%%%%%%%%%%%%%%%%%%%%%
\begin{acknowledgments}

This research was funded in whole or in part by the Austrian Science Fund (FWF) [10.55776/COE1, 10.55776/P36102].

\end{acknowledgments}

%%%%%%%%%%%%%%%%%%%%%%%%%%%%%%%%%%%%%%%%%%%%%%%%%%%%%%%%%%%%%%%%%%%%%%%%%%%%%%%%%%%%%%%%%%%%%%
\appendix
%%%%%%%%%%%%%%%%%%%%%%%%%%%%%%%%%%%%%%%%%%%%%%%%%%%%%%%%%%%%%%%%%%%%%%%%%%%%%%%%%%%%%%%%%%%%%

\section{Proof of the isomorphism $\mathcal{W}_d \cong \ZZ_d^2$ in \eqref{eq:W-H-group_projective}} \label{app:proof_isomorphism}

The \textit{Weyl-Heisenberg group} is given by \cite{bengtsson_geometry_2006}
\begin{align}
    \mathcal{V}_d := \{\omega^m W_{k,l}\}_{m,k,l\in\ZZ_d} \;,
\end{align}
where $\omega=e^{2\pi i/d}$, and $W_{k,l}$ as in \eqref{eq:Wkl}.
We want to identify elements in $\mathcal{V}_d$ that only differ up to a phase.
For this, we define an equivalence relation between Weyl operators by
\begin{align} \label{eq:equivalence_relation}
    W_{k,l} \sim W_{k',l'} \; \Longleftrightarrow \; \exists m\in\ZZ_d : W_{k,l} = \omega^m W_{k',l'},
\end{align}
with the equivalence class of $W_{k,l}$ denoted by
\begin{align} \label{eq:equivalence_classes}
    [W_{k,l}] = \{\omega^m W_{k,l}\}_{m\in\ZZ_d} \;.
\end{align}
Using the central phase subgroup $\mathcal{N}_d = \{\omega^m \id_d\}_{m\in\ZZ_d}$ of $\mathcal{V}_d$, we can write this as $[W_{k,l}] = W_{k,l} \, \mathcal{N}_d$.
Equivalently, for $U,V\in\mathcal{V}_d$ we identify $U\sim V$ iff there exists $m\in\ZZ_d$ such that $U=\omega^m V$, i.e., iff $UV^{-1}\in \mathcal{N}_d$.
The \textit{projective Weyl-Heisenberg group} is defined as the factor group
\begin{align}
    \mathcal{W}_d := \mathcal{V}_d / \mathcal{N}_d = \{[W_{k,l}] \}_{k,l\in\ZZ_d}\;.
\end{align}
Note the slight abuse of notation in \eqref{eq:W-H-group_projective}: The elements of $\mathcal{W}_d$ are not matrices $\Wkl$ but rather equivalence classes $[\Wkl]$.

In the following, we show that $\mathcal{W}_d \cong \ZZ_d^2$ as claimed in \eqref{eq:W-H-group_projective}.
The multiplication law in $\mathcal{W}_d$ is given by
\begin{align}
    [W_{k,l}] \cdot [W_{k',l'}] := [W_{k,l} W_{k',l'}] \;.
\end{align}
This definition is well-defined: if we replace the representatives by $\omega^m W_{k,l}$ and $\omega^{m'} W_{k',l'}$ with $m,m'\in\ZZ_d$, then $[(\omega^m W_{k,l})(\omega^{m'} W_{k',l'})]=[\omega^{m+m'}\,W_{k,l}W_{k',l'}]=[W_{k,l}W_{k',l'}]$, since overall phases are identified in $\mathcal{W}_d$.
Using the Weyl relations \eqref{eq:Weyl_relations} and \eqref{eq:equivalence_classes}, this yields
\begin{align}
    [W_{k,l}]\cdot [W_{k',l'}] &= [\omega^{lk'} W_{k+k',l+l'}] = [W_{k+k',l+l'}] \;.
\end{align}
Thus, multiplication in the projective group corresponds to the addition of the indices in $\ZZ_d^2$.
Now consider the map
\begin{align}
    \varphi : \ZZ_d^2 \rightarrow \mathcal{W}_d , \quad \varphi (k,l) := [W_{k,l}].
\end{align}
First, note that this map is a homomorphism. For any $(k,l),(k',l') \in \ZZ_d^2$,
\begin{align}
    \varphi\left((k,l)+(k',l')\right) &= \varphi(k+k',l+l') \\
    &= [W_{k+k',l+l'}] \\
    &= [W_{k,l}] \cdot [W_{k',l'}] \\
    &= \varphi(k,l) \cdot \varphi(k',l') \; ,
\end{align}
so $\varphi$ respects the group operation.
Second, to check injectivity, suppose $\varphi(k,l) = \varphi(k',l')$. Then $[W_{k,l}] = [W_{k',l'}]$, i.e., $W_{k,l} = \omega^m W_{k',l'}$ for some $m\in\ZZ_d$.
Taking the Hilbert–Schmidt inner product and using the orthogonality of the Weyl operators, i.e., $\mytr(W_{k,l}^\dagger W_{k',l'}) = d \, \delta_{k,k'} \delta_{l,l'}$, we find:
\begin{align}
    \mytr(W_{k,l}^\dagger W_{k',l'}) &= \mytr(\omega^{-m} W_{k',l'}^\dagger W_{k',l'}) \\
    &= \omega^{-m} \mytr(\id_d) \\
    &= \omega^{-m} d \neq 0 \; \Rightarrow \; (k,l) = (k',l') \;.
\end{align}
Hence $\varphi$ is injective.
Third, surjectivity of $\varphi$ follows directly from the definition: any element of $\mathcal{W}_d$ is by definition some equivalence class $[W_{k,l}]$, and thus lies in the image of $\varphi$.
Therefore, $\varphi$ is a bijective group homomorphism, i.e., an isomorphism:
\begin{align}
    \mathcal{W}_d \cong \ZZ_d^2 \;.
\end{align}

%%%%%%%%%%%%%%%%%%%%%%%%%%%%%%%%%%%%%%%%%%%%%%%%%%%%%%%%%%%%%%%%%%%%%%%%%%%%%%%%%%%%%%%%%%%%%
\section{\texorpdfstring{Group Structure of $\mathcal{M}_3$}{Group Structure of M3}}
\label{app:group_structure_M_3}

The 3-dimensional projective Weyl-Heisenberg group is given by $\mathcal{W}_3 = \{\Wkl\}_{k,l\in\ZZ_3} \cong \ZZ_3^2$.
As $d=3$ is prime, all subgroups $S_{k,l}\in\mathcal{S}_3$ of $\ZZ_3^2$ have $d$ elements and are cyclic, i.e., $S_{k,l}=\{n \cdot(k,l)\}_{n\in\ZZ_3}$ for some $(k,l)\in\ZZ_3^2$ with $(k,l)\neq (0,0)$.
In particular, they are given by
\begin{align}
    S_{0,1} = \{(0,0),(0,1),(0,2)\} \;, \label{eq:S(0,1)}\\
    S_{1,0} = \{(0,0),(1,0),(2,0)\} \;,\\
    S_{1,1} = \{(0,0),(1,1),(2,2)\} \;,\\
    S_{1,2} = \{(0,0),(1,2),(2,1)\} \;, \label{eq:S(1,2)}
\end{align}
with the corresponding cosets $\mathcal{C}(S_{k,l})$
\begin{align}
    \mathcal{C}(S_{0,1}) &= \{\{(0,0),(0,1),(0,2)\}, \notag \\
        &\hspace{6.5mm}\{(1,0),(1,1),(1,2)\}, \notag \\
        &\hspace{6.5mm}\{(2,0),(2,1),(2,2)\}\} \;, \label{eq:C(S_0,1)}\\
    \mathcal{C}(S_{1,0}) &= \{\{(0,0),(1,0),(2,0)\}, \notag \\
        &\hspace{6.5mm}\{(0,1),(1,1),(2,1)\}, \notag \\
        &\hspace{6.5mm}\{(0,2),(1,2),(2,2)\}\} \;,
\end{align}
\begin{align}
    \mathcal{C}(S_{1,1}) &= \{\{(0,0),(1,1),(2,2)\}, \notag \\
        &\hspace{6.5mm}\{(0,1),(1,2),(2,0)\}, \notag \\
        &\hspace{6.5mm}\{(0,2),(1,0),(2,1)\}\} \;, \\
    \mathcal{C}(S_{1,2}) &= \{\{(0,0),(1,2),(2,1)\}, \notag \\
        &\hspace{6.5mm}\{(0,1),(1,0),(2,2)\}, \notag \\
        &\hspace{6.5mm}\{(0,2),(1,1),(2,0)\}\} \label{eq:C(S_1,2)} \;.
\end{align}
Note that each element $(i,j)\in\ZZ_3^2$ appears exactly once in each $\mathcal{C}(S_{k,l})$.
Fig.~\ref{fig:striations} shows the horizontal, vertical, diagonal, and anti-diagonal striations corresponding to $\mathcal{C}(S_{0,1})$, $\mathcal{C}(S_{1,0})$, $\mathcal{C}(S_{1,1})$, and $\mathcal{C}(S_{1,2})$, respectively.
The set of all cosets $\mathcal{C}_3$ containing all possible lines in $\ZZ_3^2$ is defined in \eqref{eq:set_all_cosets} as the union of \eqref{eq:C(S_0,1)}-\eqref{eq:C(S_1,2)}.
Using \eqref{eq:subgroup_state}, the separable subgroup states in $\mathcal{M}_3$ corresponding to the striation $\mathcal{C}(S_{1,0})$ (similarly for $S_{0,1}$, $S_{1,1}$, and $S_{1,2}$) are
\begin{align}
    \rho_{1,0}^{(0)} = \frac{1}{3} (P_{0,0} + P_{1,0} + P_{2,0}) \;, \label{eq:rho_10_0} \\
    \rho_{1,0}^{(1)} = \frac{1}{3} (P_{0,1} + P_{1,1} + P_{2,1}) \;, \\
    \rho_{1,0}^{(2)} = \frac{1}{3} (P_{0,2} + P_{1,2} + P_{2,2}) \;. \label{eq:rho_10_2}
\end{align}

The projectors onto the subgroup subspaces are $\Pi(S_{k,l}) = \{3 \rho_{k,l}^{(0)}, 3 \rho_{k,l}^{(1)}, 3 \rho_{k,l}^{(2)}\}$.
This constitutes a projective measurement for each $S_{k,l}$, which, as we explain in the following, can be implemented using coarse-graining of local product basis measurements.
We first discuss this for $\Pi(S_{1,0})$ before stating the general result for $\Pi(S_{k,l})$.
We begin by defining
\begin{align} \label{eq:proj_10}
    \Pi_{k,l}^{(r)} := 3 \, \rho_{k,l}^{(r)} \;,
\end{align}
so that the projective measurement can be written as $\Pi(S_{1,0}) = \{\Pi_{1,0}^{(0)},\Pi_{1,0}^{(1)},\Pi_{1,0}^{(2)}\}$.
Evaluating \eqref{eq:proj_10} using \eqref{eq:rho_10_0}-\eqref{eq:rho_10_2} and $\Pkl := |\Omega_{k,l}\rangle\langle\Omega_{k,l}|$ with $|\Omega_{k,l}\rangle$ as in \eqref{eq:Omegakl}, we find
\begin{align} \label{eq:proj_10_decomposition}
    \Pi_{1,0}^{(r)} = \sum_{m=0}^{2} |m\rangle\langle m| \otimes |m+r\rangle\langle m+r| \;.
\end{align}
The local measurement procedure for $\Pi(S_{1,0})$ is as follows:
To measure a state $\rho\in\mathcal{M}_3$ with $\Pi(S_{1,0})$, we need to find the probabilities $\mytr(\rho \, \Pi_{1,0}^{(r)})$ for each $r\in\{0,1,2\}$, which are given by
\begin{align}
    \mytr(\rho \, \Pi_{1,0}^{(r)}) = \sum_{k=0}^{2} c_{k,r} = \sum_{(k,l)\in S_{1,0}(0,r)} c_{k,l} \;,
\end{align}
where $S_{1,0}(0,r) = \{\{0,r\},\{1,r\},\{2,r\}\}$ denotes the coset of $S_{1,0}$ for $(0,r)\in\ZZ_3^2$.
Assume both parties locally measure $\rho$ in the computational basis $\{|m\rangle\}_{m=0}^{2}$, obtaining the measurement results $m_A$ and $m_B$, respectively.
Grouping, i.e., coarse-graining, the nine possible outcomes according to $r=m_B - m_A (\mathrm{mod}\; 3)$ yields the probabilities for measuring $\Pi(S_{1,0})$:
\begin{gather}
    \begin{aligned}
        \mytr(\rho \; \Pi_{1,0}^{(r)}) &= \sum_{m=0}^{2} \mytr\left(\rho \; (|m\rangle\langle m| \otimes |m+r\rangle\langle m+r|)\right)\\
        &= \sum_{m=0}^{2} \mathrm{Pr}(m_A = m, \, m_B = m+r) \;,
    \end{aligned}
\end{gather}
where $\mathrm{Pr}(m_A = m, \, m_B = m+r)$ denotes the probability for the outcomes $m_A=m$ and $m_B = m+r$.
Hence, measuring $\Pi(S_{1,0})$ can be implemented using only local projective measurements.
We close in noting that for general $(k,l)\in\ZZ_3^2$, the elements $\Pi_{k,l}^{(r)}$ in $\Pi(S_{k,l})=\{\Pi_{k,l}^{(0)},\Pi_{k,l}^{(1)},\Pi_{k,l}^{(2)}\}$ admit a decomposition similar to \eqref{eq:proj_10_decomposition}, and can thus be measured using only local measurements and coarse-graining.
In fact, it can be checked that
\begin{align}
    \Pi_{k,l}^{(r)} = \sum_{m=0}^{2} |\psi_{k,l,m}^\ast\rangle \langle \psi_{k,l,m}^\ast| \otimes |\psi_{k,l,m+r}\rangle \langle \psi_{k,l,m+r}| \;,
\end{align}
where $|\psi_{k,l,m}^\ast\rangle$ denotes the complex conjugate of $|\psi_{k,l,m}\rangle$, and
\begin{widetext}
\begin{align}
    |\psi_{1,0,0} \rangle &= \begin{pmatrix} 1 \\ 0 \\ 0 \end{pmatrix} \;,\quad 
    |\psi_{1,0,1}\rangle = \begin{pmatrix} 0 \\ 1 \\ 0 \end{pmatrix} \;, \quad 
    |\psi_{1,0,2}\rangle = \begin{pmatrix} 0 \\ 0 \\ 1 \end{pmatrix} \;,\\
    |\psi_{0,1,0} \rangle &= \frac{1}{\sqrt{3}}\begin{pmatrix} 1 \\ 1 \\ 1 \end{pmatrix} \;,\quad 
    |\psi_{0,1,1}\rangle = \frac{1}{\sqrt{3}}\begin{pmatrix} 1\\\omega\\\omega^2 \end{pmatrix} \;, \quad 
    |\psi_{0,1,2}\rangle = \frac{1}{\sqrt{3}}\begin{pmatrix} 1\\\omega^2\\\omega \end{pmatrix} \;,\\
    |\psi_{1,1,0} \rangle &= \frac{1}{\sqrt{3}}\begin{pmatrix} 1\\\omega^2\\\omega^2 \end{pmatrix} \;,\quad 
    |\psi_{1,1,1}\rangle = \frac{1}{\sqrt{3}}\begin{pmatrix} 1\\\omega\\1 \end{pmatrix} \;, \quad 
    |\psi_{1,1,2}\rangle = \frac{1}{\sqrt{3}}\begin{pmatrix} 1\\1\\\omega \end{pmatrix} \;,\\
    |\psi_{1,2,0} \rangle &= \frac{1}{\sqrt{3}}\begin{pmatrix} 1\\\omega\\\omega \end{pmatrix} \;,\quad 
    |\psi_{1,2,1}\rangle = \frac{1}{\sqrt{3}}\begin{pmatrix} 1\\\omega^2\\1 \end{pmatrix} \;, \quad 
    |\psi_{1,2,2}\rangle = \frac{1}{\sqrt{3}}\begin{pmatrix} 1\\1\\\omega^2 \end{pmatrix} \;,
\end{align}
\end{widetext}
where $\omega= e^{2\pi i/d}$.
Note that the sets $\mathcal{B}_{k,l}:=\{|\psi_{k,l,m}\rangle\}_{m=0}^{2}$ constitute bases of $\C^3$ for each $(k,l)\in\ZZ_3^2$, and $\{\mathcal{B}_{k,l}\}_{(k,l)\in\ZZ_3^2}$ is a complete set of four mutually unbiased bases of $\C^3$.

%%%%%%%%%%%%%%%%%%%%%%%%%%%%%%%%%%%%%%%%%%%%%%%%%%%%%%%%%%%%%%%%%%%%%%%%%%%%%%%%%%%%%%%%%%%%%
\section{Proof of Proposition~\ref{thm:no_PPT-ENT_with_one_line_zeros}} \label{app:proof_thm_lines_BE}

The subgroups of $\ZZ_3^2$ are cyclic and correspond to vertical, horizontal, diagonal, and anti-diagonal lines in the affine plane $\ZZ_3^2$ (cf. Fig.~\ref{fig:striations}), or equivalently in the coefficient matrix $C$ in \eqref{eq:weight_matrix_d3}.
Hence, we need to show that if the coefficients $\ckl$ (satisfying $\ckl\geq0$ and $\sum_{k,l=0}^{2}\ckl=1$) corresponding to any coset $\ell\in\mathcal{C}_3$ are zero, the corresponding state $\rho\in\mathcal{M}_3$ is either separable or NPT entangled.
To check the latter, we can use Thm.~\ref{thm:PPT_crit_BDS}.
Due to the symmetries of $\mathcal{M}_3$ that preserve the entanglement class of $\rho\in\mathcal{M}_3$ \cite{baumgartner_special_2007, baumgartner_state_2006}, this reduces to six distinct cases.
There is a total of $432=|\mathrm{AGL}(2,3)|$ such symmetry transformations that map cosets to cosets (cf. Sec.~\ref{sec:bloch_vec_and_DFT}).
In the following, each case beyond the first is obtained from the former by adding one additional $\ckl\neq 0$ and accounting for these symmetries.

\textbf{Case 1:}
If only one $\ckl > 0$, it is necessarily equal to 1 (as $\sum_{k,l=0}^{2} \ckl=1$).
Utilizing the symmetries of $\mathcal{M}_3$, we can shift any one $\ckl$ to $c_{0,0}$ by local operations.
Hence, this case reduces to
\begin{align}
    C_1 = \begin{pmatrix}
        1 & 0 & 0 \\
        0 & 0 & 0 \\
        0 & 0 & 0
        \end{pmatrix} \;,
\end{align}
for which \eqref{eq:det(A_0)} yields $0<1$, and thus $\rho \in \NPT$.

\textbf{Case 2:}
If two $\ckl > 0$, we can locally shift one coefficient $\ckl$ to $c_{0,0}$ and the other one to $c_{0,1}$, so this case reduces to
\begin{align}
    C_2 = \begin{pmatrix}
        c_{0,0} & c_{0,1} & 0\\
        0 & 0 & 0 \\
        0 & 0 & 0
        \end{pmatrix} \;.
\end{align}
Equation \eqref{eq:det(A_0)} becomes $0<c_{0,0}^3 +c_{0,1}^3$, which is satisfied.
Hence, $\rho\in\NPT$.

\textbf{Case 3:}
Three coefficients $\ckl > 0$ reduce to two distinct cases after utilizing the symmetries of $\mathcal{M}_3$.
The first one is given by
\begin{align}
    C_{3a} = \begin{pmatrix}
        c_{0,0} & c_{0,1} & 0 \\
        c_{1,0} & 0 & 0 \\
        0 & 0 & 0 
        \end{pmatrix} \;,
\end{align}
with \eqref{eq:det(A_0)} being $0< c_{0,0}^3 + c_{0,1}^3 + c_{1,0}^3$, and thus $\rho\in\NPT$.
The second one is
\begin{align}
    C_{3b} = \begin{pmatrix}
        c_{0,0} & c_{0,1} & c_{0,2} \\
        0 & 0 & 0 \\
        0 & 0 & 0
        \end{pmatrix} \;,
\end{align}
for which \eqref{eq:det(A_0)} becomes (after some simplifications)
\begin{align}
    -|z|^2 \; \underbrace{(c_{0,0} + c_{0,1} + c_{0,2})}_{=1} <0 \;, \label{eq:det_A_0_3ckl}
\end{align}
where $z:=c_{0,0} + \omega \, c_{0,1} + \omega^\ast \, c_{0,2}$ and $\omega = e^{2\pi i/3}$.
The inequality is violated if and only if $|z|=0$, which is equivalent to $c_{0,0} = c_{0,1} = c_{0,2} = 1/3$ (because $\sumkl\ckl=1$), and corresponds to a separable subgroup state \cite{baumgartner_special_2007}.
Otherwise, the inequality is satisfied and $\rho\in \NPT$.

\textbf{Case 4:}
With four coefficients $\ckl > 0$, we again distinguish two cases.
The first one is
\begin{align} \label{eq:C_4a}
    C_{4a} = \begin{pmatrix}
        c_{0,0} & c_{0,1} & c_{0,2} \\
        c_{1,0} & 0 & 0 \\
        0 & 0 & 0
        \end{pmatrix} \;.
\end{align}
Equation \eqref{eq:det(A_0)} is given by
\begin{align}
    \underbrace{3 \, c_{0,0} \, c_{0,1} \, c_{0,2} - c_{0,0}^3 - c_{0,1}^3 - c_{0,2}^3}_{\leq0 \text{ by \eqref{eq:det_A_0_3ckl}}} \; \underbrace{- c_{1,0}^3}_{<0} <0 \;,
\end{align}
and thus $\rho\in\NPT$.
The second case is
\begin{align} \label{eq:C_4b}
    C_{4b} = \begin{pmatrix}
        c_{0,0} & c_{0,1} & 0 \\
        c_{1,0} & c_{1,1} & 0 \\
        0 & 0 & 0
        \end{pmatrix} \;,
\end{align}
yielding $0< c_{0,0}^3+c_{0,1}^3 + c_{1,0}^3 + c_{1,1}^3$, and thus $\rho\in\NPT$.

\textbf{Case 5:}
The case with five $\ckl > 0$ is obtained by adding one extra $\ckl$ to either \eqref{eq:C_4a} or \eqref{eq:C_4b}.
After accounting for the symmetries in $\mathcal{M}_3$, the only possibility with one complete line of zeros is
\begin{align}
    C_{5} = \begin{pmatrix}
        c_{0,0} & c_{0,1} & c_{0,2} \\
        c_{1,0} & c_{1,1} & 0 \\
        0 & 0 & 0
        \end{pmatrix} \;.
\end{align}
This gives
\begin{align}
    \underbrace{3 \, c_{0,0} \, c_{0,1} \, c_{0,2} - c_{0,0}^3 - c_{0,1}^3 - c_{0,2}^3}_{\leq0 \text{ by \eqref{eq:det_A_0_3ckl}}} \;\underbrace{- c_{1,0}^3 - c_{1,1}^3}_{<0} <0 \;,
\end{align}
and thus $\rho\in\NPT$.

\textbf{Case 6:}
Accounting for the symmetries of $\mathcal{M}_3$, the only possibility to have six $\ckl > 0$ but still one full line equal to zero is given by
\begin{align}
    C_{6} = \begin{pmatrix}
        c_{0,0} & c_{0,1} & c_{0,2} \\
        c_{1,0} & c_{1,1} & c_{1,2} \\
        0 & 0 & 0
        \end{pmatrix} \;,
\end{align}
for which \eqref{eq:det(A_0)} is
\begin{align}
    &\underbrace{3 \, c_{0,0} \, c_{0,1} \, c_{0,2} - c_{0,0}^3 - c_{0,1}^3 - c_{0,2}^3}_{\leq0 \text{ by \eqref{eq:det_A_0_3ckl}}} \; + \notag\\
    &\hspace{6mm} + \, \underbrace{3 \, c_{1,0} \, c_{1,1} \, c_{1,2} - c_{1,0}^3 - c_{1,1}^3 - c_{1,2}^3}_{\leq0 \text{ by \eqref{eq:det_A_0_3ckl}}} <0 \;.
\end{align}
Thus, the inequality is violated if and only if $c_{0,0} = c_{0,1}= c_{0,2}$ and $c_{1,0}=c_{1,1}=c_{1,2}$, which corresponds to a convex mixture of two separable subgroup states \eqref{eq:subgroup_state}, and thus is also separable.
Otherwise, it is satisfied and $\rho\in\NPT$.

%%%%%%%%%%%%%%%%%%%%%%%%%%%%%%%%%%%%%%%%%%%%%%%%%%%%%%%%%%%%%%%%%%%%%%%%%%%%%%%%%%%%%%%%%%%%%
\section{\texorpdfstring{Explicit form of $\kappa_{i,j}$ for $W_\NPT$}{Explicit form of kappaij for WNPT}} \label{app:explicit_form_of_W_NPT}

Here, we give the explicit form of $\kappa_{i,j}$ in \eqref{eq:kappa_ij} for the hyperplanes \eqref{eq:W_NPT}:

\begin{align}
    \kappa_{0,0} &= -c_{0,0}^2 + c_{0,1} c_{0,2} + c_{1,0} c_{2,0} + c_{1,1} c_{2,2} + c_{1,2} c_{2,1} \;, \notag\\
    \kappa_{0,1} &= -c_{0,1}^2 + c_{0,0} c_{0,2} + c_{1,1} c_{2,1} + c_{1,2} c_{2,0} + c_{1,0} c_{2,2} \;, \notag\\
    \kappa_{0,2} &= -c_{0,2}^2 + c_{0,0} c_{0,1} + c_{1,2} c_{2,2} + c_{1,0} c_{2,1} + c_{1,1} c_{2,0} \;, \notag\\
    \kappa_{1,0} &= -c_{1,0}^2 + c_{1,1} c_{1,2} + c_{0,0} c_{2,0} + c_{0,2} c_{2,1} + c_{0,1} c_{2,2} \;, \notag\\
    \kappa_{1,1} &= -c_{1,1}^2 + c_{1,0} c_{1,2} + c_{0,1} c_{2,1} + c_{0,0} c_{2,2} + c_{0,2} c_{2,0} \;, \notag\\
    \kappa_{1,2} &= -c_{1,2}^2 + c_{1,0} c_{1,1} + c_{0,2} c_{2,2} + c_{0,1} c_{2,0} + c_{0,0} c_{2,1} \;, \notag\\
    \kappa_{2,0} &= -c_{2,0}^2 + c_{2,1} c_{2,2} + c_{0,0} c_{1,0} + c_{0,1} c_{1,2} + c_{0,2} c_{1,1} \;, \notag\\
    \kappa_{2,1} &= -c_{2,1}^2 + c_{2,0} c_{2,2} + c_{0,1} c_{1,1} + c_{0,2} c_{1,0} + c_{0,0} c_{1,2} \;, \notag\\
    \kappa_{2,2} &= -c_{2,2}^2 + c_{2,0} c_{2,1} + c_{0,2} c_{1,2} + c_{0,0} c_{1,1} + c_{0,1} c_{1,0} \;. \notag
\end{align}
Here, the second, third, fourth, and fifth term in each $\kappa_{i,j}$ corresponds to the horizontal, vertical, diagonal, and anti-diagonal cosets containing $c_{i,j}$, respectively.

% %%%%%%%%%%%%%%%%%%%%%%%%%%%%%%%%%%%%%%%%%%%%%%%%%%%%%%%%%%%%%%%%%%%%%%%%%%%%%%%%%%%%%%%%%%%%%
\bibliography{references}% Produces the bibliography via BibTeX.

\end{document}